\pdfoutput=1
\documentclass[10pt,pra,aps,superscriptaddress,twocolumn]{revtex4-1}
\usepackage{amsmath,bbm}
\usepackage{amsthm}
\usepackage{amsmath}
\usepackage{latexsym}
\usepackage{amssymb}
\usepackage{float}
\usepackage{comment}
\usepackage{braket}
\usepackage{graphicx}           
\usepackage{color}
\usepackage{mathpazo}
\usepackage{comment}
\usepackage{enumitem}
\usepackage{multirow}
\usepackage{setspace}
\usepackage{xcolor}
\usepackage[colorlinks=true,linkcolor=blue,citecolor=magenta,urlcolor=blue]{hyperref}
\usepackage{amsmath,amssymb,amsthm,bm,amsfonts,mathrsfs,bbm} 
\usepackage{physics}

\newcommand{\be}{\begin{equation}}
\newcommand{\ee}{\end{equation}}
\newcommand{\bea}{\begin{eqnarray}}
\newcommand{\eea}{\end{eqnarray}}

\def\squareforqed{\hbox{\rlap{$\sqcap$}$\sqcup$}}
\def\qed{\ifmmode\squareforqed\else{\unskip\nobreak\hfil
\penalty50\hskip1em\null\nobreak\hfil\squareforqed
\parfillskip=0pt\finalhyphendemerits=0\endgraf}\fi}
\def\endenv{\ifmmode\;\else{\unskip\nobreak\hfil
\penalty50\hskip1em\null\nobreak\hfil\;
\parfillskip=0pt\finalhyphendemerits=0\endgraf}\fi}
\newcommand{\I}{\mathbbm{1}}

\newcommand{\od}{\overline{d}}

\newcommand{\la}{\langle}
\newcommand{\ra}{\rangle}

\newcommand{\re}{\color{blue}}  
\newcommand{\blk}{\color{black}}

\makeatletter
\newtheorem*{rep@theorem}{\rep@title}
\newcommand{\newreptheorem}[2]{%
\newenvironment{rep#1}[1]{%
 \def\rep@title{#2 \ref{##1}}%
 \begin{rep@theorem}}%
 {\end{rep@theorem}}}
\makeatother

\newtheorem{theorem}{Theorem}%[section]
\newreptheorem{theorem}{Theorem}

\newtheorem{definition}{Definition}

\newtheorem{result}{Observation}

\newtheorem{coro}{Corollary}

\begin{document}

%%%%%%%%%%%%%%%%%%%%%%%%%%%%%%%%%%%%%%%%%%%%%%%%%%%%%%%%%%%%%%%%%%%

\title{Classifying Measurement Incompatibility under Classical Pre- and Post-Processing Operations}
%%%%%%%%%%%%%%%%%%%%%%%%%%%%%%%%%%%%%%%%%%%%%%%%%%%%%%%%%%%%%%%%%%%

\author{Arun Kumar Das}
\email{akd.qip@gmail.com}
\affiliation{S. N. Bose National Centre for Basic Sciences, Block JD, Sector III, Salt Lake, Kolkata 700106, India}

\author{Saheli Mukherjee}
\email{mukherjeesaheli95@gmail.com}
\affiliation{S. N. Bose National Centre for Basic Sciences, Block JD, Sector III, Salt Lake, Kolkata 700106, India}

\author{Debashis Saha}
\email{dsaha@iitbbs.ac.in}
\affiliation{School of Physics, Indian Institute of Science Education and Research Thiruvananthapuram, Kerala 695551, India}
\affiliation{Department of Physics, School of Basic Sciences, Indian Institute of Technology Bhubaneswar, Odisha 752050, India}

\author{Debarshi Das}
\email{dasdebarshi90@gmail.com}
\affiliation{Department of Physics, Shiv Nadar Institution of Eminence, Gautam Buddha Nagar, Uttar Pradesh 201314, India}
\affiliation{Department of Physics and Astronomy, University College London, Gower Street, WC1E 6BT London, England, United Kingdom}

\author{A. S. Majumdar}
\email{archan@bose.res.in}
\affiliation{S. N. Bose National Centre for Basic Sciences, Block JD, Sector III, Salt Lake, Kolkata 700106, India}

%%%%%%%%%%%%%%%%%%%%%%%%%%%%%%%%%%%%%%%%%%%%%%%%%%%%%%%%%%%%%%%%%%%

\begin{abstract}
Measurement incompatibility has proved to be an important resource for quantum information processing. In this work, we present an operational approach that leverages classical operations on the inputs (pre-processing) and outputs (post-processing) of measurement devices to explore different layers of incompatibility among the measurements performed by the device. We study classifications of measurement incompatibility with respect to these two types of classical operations, {\it viz.}, post-processing or coarse-graining of measurement outcomes and pre-processing or convex-mixing of different measurements.  We derive analytical criteria for determining when a set of projective measurements is fully incompatible with respect to coarse-graining or convex-mixing. Robustness against white noise for different layers of incompatibility for mutually unbiased bases is investigated. Furthermore, we study operational witnesses for incompatibility subject to these classical operations, using the input-output statistics of Bell-type experiments as well as experiments in the prepare-and-measure scenario. 
\end{abstract}

\maketitle

%%%%%%%%%%%%%%%%%%%%%%%%%%%%%%%%%%%%%%%%%%%%%%%%%%%%%%%%%%%%%%%%%%% 

\section{introduction}

Measurement incompatibility is a concept related to observables that cannot be measured jointly with arbitrary accuracy \cite{guhne2023colloquium}. It is purely a quantum effect, of which the most well-known example concerns the position and momentum of a quantum particle.  Being a fundamental concept of quantum theory, it takes a pivotal role in explaining several quantum phenomena,  such as, Bell-nonlocality \cite{Fine1982,PRL103230402}, Einstein-Podolsky-Rosen steering \cite{PhysRevLett.113.160403,PhysRevLett.113.160402,tp1,PhysRevLett.115.230402,tp2}, measurement uncertainty
relations \cite{tp_prl, mur, MDM17, saha2020}, quantum contextuality \cite{LIANG20111,Xu2019,sg_prl}, quantum violation of macrorealism \cite{pla280.2265,PhysRevA.100.042117}, and temporal
and channel steering \cite{Karthik:15,PhysRevA.91.062124,PhysRevA.97.032301}.  
Apart from the foundational significance of incompatible measurements, they have been proven to be an important resource for various information processing tasks \cite{PhysRevLett.122.130403,Buscemi_prl,Heinosaari_prl}.   Recently, It is also known that measurement incompatibility is necessary
for quantum advantage in any one-way communication task \cite{SahaIncom}. 

 The significance of measurement incompatibility in various operational tasks calls for its in-depth characterization. Towards this direction, a classification of measurement incompatibility with respect to projection onto subspaces has been recently performed \cite{uola21}, also different hierarchy of incompatible measurements has been given based on the minimum number of copies of quantum states required \cite{mathCarmeli}, and the minimum number of measurements required to simulate them \cite{measuremestSimulability_guerini,prl17_oszmaniec, filippovPra18,egelhaaf_pmBell}.
% In the continuous variable system, the verification of incompatibility using phase-space quasiprobability distributions has been studied in Ref.\cite{cv_incom}.

In the present work, our objective is to classify measurement incompatibility in an operational approach that does not involve the details of a theory. To put it differently, we address how various degrees of incompatibility can be assessed solely by executing basic classical operations on the inputs or outputs of these measurements. From an operational perspective, when the concerned measurement devices are black boxes with no control over the internal workings, one can still realize different measurements by performing suitable classical operations to manipulate the statistics of the measurements.
Here, we employ two such operations: coarse-graining of measurement outcomes (classical post-processing performed on the outcomes of a measurement) and convex-mixing of measurement settings (classical pre-processing performed on the inputs). One can also consider classical operations performed on both inputs and outputs.   Similar types of classical operations were also introduced in \cite{haapasaloQip12}, nevertheless, the study of incompatibility under these classical operations was lacking. The motivation of this present work is to fill this crucial gap in the literature.

Coarse-graining of measurement outcomes arises naturally in several instances, for example, in measurements on continuous variable systems \cite{cg_uncertainty}. Though the eigen spectra of the observables are infinite-dimensional and continuous, real-world experimental devices are limited by finite precision, leading to the measurement outcomes taking a finite number of discrete values. This inaccuracy in the recording of measurement outcomes is manifested in the coarse-graining of measurement outcomes, which is inevitable in practice. %Coarse-graining has also been employed to study the phenomenon of quantum-to-classical transition. It is observed that quantum phenomena may disappear due to imprecision of measurement outcomes \cite{prl07_brukner,rudnicki_UncertaintyCg,clQuCG_Jeong,pla280.2265,PhysRevA.94.062117,PhysRevA.100.042114}. 
On the other hand, device imperfection may also lead to the measurement device performing a set of measurements probabilistically, instead of always performing the desired particular measurement. In such a case, a convex-mixing of the given set of measurements arises effectively \cite{Ariano_2005}. 

By leveraging these classical operations on the measurement device, we can learn the finer details about the measurements i.e. the different layers of incompatibility, which can give a measure of the degree of incompatibility of a set of measurements. This can set a benchmark for legitimately choosing incompatible measurements for an information processing task. Our study is motivated to address how one may compare the degree of incompatibility between two different sets of measurements subjected to the aforementioned classical operations. For instance, if the first set remains incompatible for every possible non-trivial coarse-graining of the measurement outcomes, but the second set becomes compatible for a certain coarse-graining, it follows that the first set of measurements 
exhibits stronger incompatibility compared to the second one. A similar argument holds for the case of convex-mixing of measurements also.

% As the incompatibility of measurements is a quantum concept, it is interesting to examine how this property behaves under such elementary classical operations. {\color{red} It is quite expected that these classical operations, namely coarse-graining of outcomes and convex mixing of measurements may affect the degree of incompatibility of a given set of measurements. However, it is a priori unknown what effects these classical operations have on the incompatibility of measurements.

% Our study is motivated to address how one may compare the degree of incompatibility between two different sets of measurements subjected to the aforementioned classical operations.}  

In this work, we establish analytical criteria for determining when a pair of projective measurements are fully incompatible, i.e., remain incompatible under all possible coarse-grainings of measurement outcomes. We give the necessary and sufficient condition for the full-incompatibility with respect to coarse-graining for any two rank-one projective measurements in dimension $d$ and we provide many supportive examples to visualise this. We notice that any two mutually unbiased bases in prime dimension are fully incompatible with respect to coarse-graining. We further analyse the full incompatibility of a set of three qubit measurements under all possible convex-mixing as well. Within the context of our present study, noise is reflected in degrading the incompatibility properties of various measurement sets \cite{mate}. We compute the critical noise threshold below which the mutually unbiased bases measurements remain incompatible under the above-mentioned classical operations.

 As incompatible measurements are useful for various information processing tasks, any device claiming to produce incompatible measurements must be certified before using it in an experiment. Verification of incompatibility of the measurements is possible
from the input-output measurement statistics obtained from the device 
without knowing its internal functioning. This can be done in two ways one is a device-independent way utilizing Bell-type experiments \cite{wolf_prl,Bene_2018,PhysRevA.97.012129} and another is a semi-device-independent way inspired by the standard prepare-and-measure scenario \cite{PRXQuantum_carlos,SahaIncom,Carmeli2020}.  
% Device-independent protocols are conceptually most powerful, relying only on the input-output statistics  \cite{review_bell},
% with  a wide range of applications \cite{ekertQkd,diQkd,revCommCompl,selfTest_mayers, SupicSelftest20,suche_pra18, bian_pra,prl19_diIncom}.  However, they require shared entanglement, an expensive resource, and prohibition of communication between the involved parties. 
% On the other hand,  semi-device-independent protocols are inspired by the standard prepare-and-measure scenario, with an additional constraint on the dimension of the communicated quantum states \cite{qkd_pawlowski,quntumNetwork_bowles,randomnessPassaro_2015,racPawlowski,physreva.103.062604, quantum.6.716,VanHimbeeck2017semidevice, semiDeviceInd_moreno}.
In the present work, we investigate the issue of certification of different layers of measurement incompatibility 
under the introduced classical operations both from the device-independent and the semi-device-independent perspective. As both Bell-type experiments and prepare-and-measure experiments \cite{cglmp_expt_pra24,armin_prl15} are viable with the current technology, by performing our introduced classical operations on those experiments we can operationally certify the different layers of incompatibility thus allowing us to know the subtleties of incompatibility about the measurement device which is a new and interesting offshoot of our work.  %We study incompatibility witnesses in the device-independent and the semi-device-independent scenarios. 
 % {\color{teal} Moreover, our approach of witnessing incompatibility subjected to these classical operations in a device-independent scenario is viable for experimental implementation and hence, one can obtain more information about these layers of incompatibility from previously conducted experiments \citep{cglmp_expt_pra24}.}

Our paper is structured as follows. In Sec.(\ref{sec2}), the significance of these classical operations in an operational paradigm is discussed along with the physical ground for these classical operations. In Sec. (\ref{sec4}), we define a hierarchy in the incompatibility of measurements under these classical operations. In Sec. (\ref{sec5}), we study the noise robustness
for various levels of incompatibility of measurements subjected to the above classical operations. In Sec. (\ref{sec6}), we define operational
witnesses of incompatibility of measurements under these classical operations and furnish examples to study their
performance in device-independent and semi-device-independent frameworks. Concluding remarks are presented in Sec.(\ref{sec7}).

\section{Classifying incompatibility under classical operations of the inputs and the outputs} \label{sec2}

Consider a measurement device where, when a physical system is probed, we can choose measurement settings via a classical parameter, denoted by $x$. This parameter serves as the input to the device, and the device produces the measurement outcome $z$, corresponding to that input. Suppose there are $n$ possible inputs, i.e., $x \in [n]$, and for each input, there are $d$ possible outcomes, i.e., $z \in [d]$, where $[k]$ represents the set $\{0, \cdots, k-1\}$ for any natural number $k$. For simplicity, we assume that the number of outcomes for each measurement is the same. Without loss of generality, this assumption holds, as any measurement with fewer than $d$ outcomes can be treated as having $d$ outcomes by considering some outcomes as never occurring. By probing different system preparations, we can gather input-output statistics from the device.

In quantum mechanics, measurements are described by Positive Operator Valued Measures (POVMs), which consist of positive semi-definite operators that sum up to the identity operator. Let us represent the measurements realized in the device by the set of operators $\{M_{z|x}\}_{z,x}$, where $x$ indexes different measurements, and $z$ denotes the corresponding outcomes. A set of measurements is said to be compatible if there exists a parent POVM, $G_\lambda$, and classical post-processing $\{p(z|x,\lambda)\}$ for each $x$ such that
\be \label{def_comp}
\forall z,x, \quad M_{z|x} = \sum_\lambda p(z|x,\lambda) G_\lambda ,
\ee 
where $0\le p(z|x,\lambda)\le 1,$ and $\sum_{z}p(z|x,\lambda)=1,$ for all $x,\lambda$ \cite{guhne2023colloquium}. 
As a special case, if the operators are projectors, then the two measurements are jointly measurable when their corresponding operators commute.

In an operational paradigm, we do not have direct control over the internal workings of the measurement device. However, we can manipulate the classical inputs and outputs to realize different measurements. 

\begin{figure}[h!]
    \centering
    \includegraphics[width=1\linewidth]{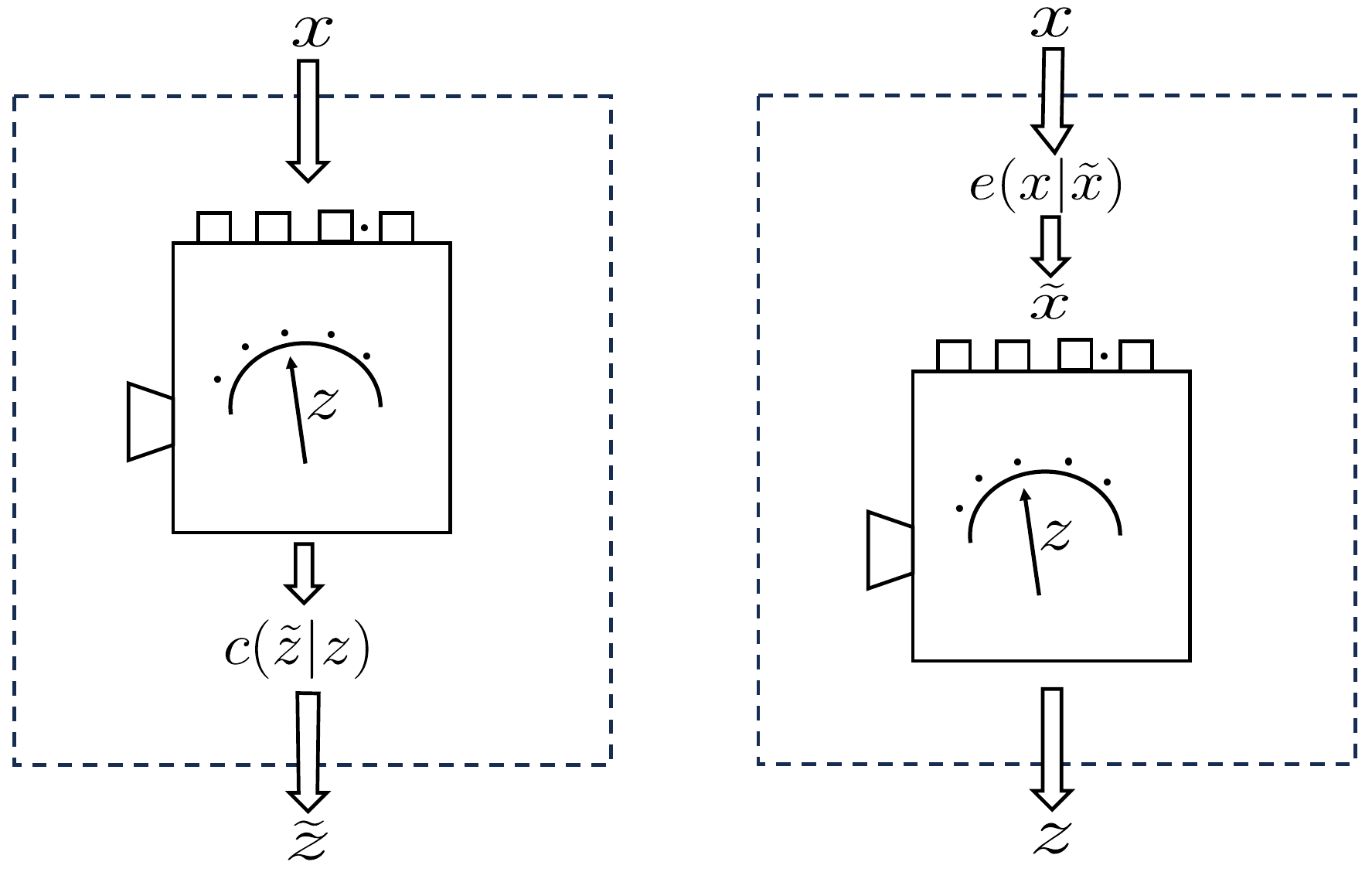}
    \caption{Classical operations on the outputs (left) and classical operations on the inputs (right) to realize a new set of measurements are depicted.}
    \label{fig}
\end{figure}

\subsection{Classical operations on outputs}

In general, one can modify the outcomes of a measurement by applying a post-processing operation or stochastic map (see FIG. \ref{fig}). For a given input $x$, this operation is mathematically described by a probability distribution $\mathcal{C}_x := \{c_x(\tilde{z}_x|z)\}$, $c_x(\tilde{z}_x|z) \geqslant 0$ and $\sum_{\tilde{z}_x} c_x(\tilde{z}_x|z) =1$. Here $\tilde{z}_x \in [\tilde{d}_x]$ denotes the outcome of the new measurement after the operation, and $\mathcal{C}_x$ is the stochastic map with $c_x(\tilde{z}_x|z)$ is the probability of producing outcome $\tilde{z}_x$ when outcome $z$ is obtained for the measurement $x$. In general, one can apply different operations for different $x$. By applying different operations for each $x$, one can generate a new set of measurements described by the operators $\{\tilde{M}_{\tilde{z}_x|x}\}_{\tilde{z}_x,x}$ as follows: 
\be \label{co-in}
\tilde{M}_{\tilde{z}_x|x} = \sum_{z} c_x(\tilde{z}_x|z) M_{z|x} \quad 
,
\ee 
for each $x.$ If the initial set of measurements is incompatible, we can ask whether this new set is incompatible or not. 

The set of all possible operations $\{\mathcal{C}_x\}$ is infinite. Furthermore, certain operations, such as generating random outputs, will always result in a set of measurements that is compatible, regardless of the original measurements. To avoid such intricacies, we focus on a particular class of operations where $c_x(\tilde{z}_x|z) \in \{0,1\}$ and $\tilde{d}_x < d$. These classical operations involve coarse-graining or relabelling/permuting outcomes (or a combination of both). This class contains a finite number of operations, and interestingly, any classical operation on outputs can be expressed as convex combinations of this class of operations. More importantly, certain quantum measurements remain incompatible even after all possible operations of this kind. We will refer to this class of operations as \textit{coarse-graining}. 
Coarse-graining typically refers to treating more than one outcome as equivalent, {\it i.e.},  multiple outcomes are clubbed together to form a single outcome. As a result, the effective number of outcomes is reduced. 
%In this work, we include permutation or relabelling of outputs into it. 

\subsection{Classical operation on inputs}

We can also apply classical operations to the inputs to generate a new set of measurements, each of which is a convex combination of the initial measurements (see FIG. \ref{fig}). This classical operation is represented by a probability distribution $\{e(x|\tilde{x})\}_{x,\tilde{x}}$, where $e(x|\tilde{x})\geqslant 0$ and $\sum_{x} e(x|\tilde{x})=1$ for each $\tilde{x} \in [m]$. Here, $e(x|\tilde{x})$ refers to the convex weightage of appearing measurement $x$ in the new measurement labelled by $\tilde{x}$. The new set of $m$ measurements, defined by the operators $\{\tilde{M}_{z|\tilde{x}}\}_{z,\tilde{x}}$, is given by:
\be
\tilde{M}_{z|\tilde{x}} = \sum_{x} e(x|\tilde{x}) M_{z|x} \quad ,
\ee 
for every $z$.

%If the same convex mixing is applied uniformly to all inputs, i.e., $e(x|\tilde{x})$ is the same for all $\tilde{x}$, the resulting measurements will be identical and hence compatible. 
It is straightforward to see that if we allow all possible convex mixtures of measurements, the resulting measurements can become identical, regardless of the initial set.
Consider two convex mixtures that produce two new measurements from the a set of initial measurements. If at least one initial measurement can appear in both mixtures, we can choose the mixture weights so that this particular measurement has weight 1 (or close to 1) in both cases. As a result, the two new measurements will be identical (or nearly identical), making them compatible irrespective of the form of the initial measurement set. To avoid such trivial cases, we introduce a restricted class of convex mixing, which we will refer to as \textit{disjoint-convex-mixing} of measurements.
In disjoint-convex-mixing, each initial measurement is allowed to appear in the convex combination for only one new measurement, ensuring that no initial measurement contributes to more than one of the resulting measurements. Mathematically, for every $x$, only one of the numbers from the set $\{e(x|\tilde{x})\}_{\tilde{x}}$ is non-zero, and others are zero, which can be expressed as:
\be 
\forall x, \ \max_{\tilde{x}} \{e(x|\tilde{x}) \} =  \sum_{\tilde{x}} e(x|\tilde{x}) .
\ee 
In this work, we focus on disjoint-convex-mixing for classifying measurement incompatibility. We will point out later that there are measurements that remain incompatible after all possible disjoint-convex-mixing.

\subsection{Classical operations on both inputs and outputs}

It is also possible to realize new measurements by first taking convex-mixing on the inputs, followed by coarse-graining on the outputs.
In this case, the operators $\{\tilde{M}_{\tilde{z}_{\tilde{x}}|\tilde{x}}\}_{\tilde{z}_{\tilde{x}},\tilde{x}}$ defining the new measurements are
\be 
\tilde{M}_{\tilde{z}_{\tilde{x}}|\tilde{x}} = \sum_{z}\sum_{x} c(\tilde{z}_x|z) e(x|\tilde{x}) M_{z|x} \quad ,
\ee 
where $\{e(x|\tilde{x})\}$ and $\{c(\tilde{z}_x|z)\}_x$ denote the classical operation on inputs and outputs, respectively.
% {\color{teal} This is discussed explicitly in Sec.\ref{sec5} with supporting examples. }

%In the latter case, the set of operators $\{\tilde{M}_{\tilde{z}|\tilde{x}}\}$ are\be \tilde{M}_{\tilde{z}|\tilde{x}} = \sum_{x}\sum_{z} e(x|\tilde{x}) c(\tilde{z}|z)  M_{z|x} \quad ,\ee where $\{e(x|\tilde{x})\}$ and $\{c(\tilde{z}|z)\}$ denote the classical operation on inputs and outputs, respectively. Here again, without loss of generality, we consider the number of outcomes after the operations on the outcomes to be the same, which is denoted by $\tilde{z}$. It is important to note that these two cases are not equivalent, as the stochastic operations on the outputs apply to different sets of measurements in each case. \\
%For example, a set of measurements realized by first permuting the outcomes of measurement followed by convex-mixing cannot be realized in the former case where we first perform convex-mixing and then perform classical operations of outputs. \\ 

In summary, one can explore the effects of various classical operations on both inputs and outputs to determine whether the resulting measurements remain incompatible. However, this work specifically focuses on two key classes of operations: coarse-graining of outputs and disjoint-convex-mixing of inputs. In addition to the operational relevance of the two types of operations mentioned above, their significance for practical implementation is discussed in the following subsection.

 Recently, the concept of measurement simulability was introduced in  \cite{mathCarmeli,measuremestSimulability_guerini,prl17_oszmaniec} -- given a set of $m$ measurements, which other set of $n$ measurements (with $n < m$)  can simulate the initial $m$ measurements using classical pre- and post-processing on the $n$ measurements. Note that this definition generalizes measurement incompatibility, since for $n=1$ it reduces to the standard concept of measurement incompatibility. On the other hand, the present paper investigates the following: Given a set of measurements, we generate a new set by applying classical pre- and post-processing. We then ask whether there exists a single parent POVM that can simulate this new set of measurements via classical post-processing.

\subsection{Physical motivations for considering the two specific classical operations}

Let us now elaborate further the motivations for considering the two aforementioned specific classical operations on a set of measurements. Coarse-graining of measurement outcomes is a natural consequence of an observer's limitation in a practical scenario involving multi-outcome measurements. For example, consider the measurement of the spin-$z$ component (associated with the operator $J_z$) of the spin-$j$ system, where $j$ is very large. The possible outcomes of this measurement are nothing but the eigenvalues of $J_z$, denoted by $m \in \{-j, -j+1, \cdots, j\}$. In practice, this type of measurement is performed using a Stern-Gerlach type experiment. In such a practical situation, the concept of ``neighbouring outcomes'' arises \cite{prl07_brukner}. For example, the outcomes $m$ and $m+1$ are neighbouring outcomes in the real configuration space in a Stern-Gerlach type experiment. For large $j$, it is almost impossible for a device with finite precision to resolve these neighbouring outcomes in the observation screen -- giving rise to coarse-graining of measurement outcomes. This type of practical limitation is not only limited to large spins but also applicable to other multi-outcome measurements and measurements of continuous variables (e.g., position, momentum \cite{PhysRevLett.132.030202}). Note that coarse-graining of measurement outcomes has been invoked in the context of explaining classical limits of quantum mechanics \cite{prl07_brukner,rudnicki_UncertaintyCg,clQuCG_Jeong,pla280.2265,PhysRevA.94.062117,PhysRevA.100.042114}.

On the other hand, convex mixing of measurements is also inevitable in a practical situation. To explain it in more detail, let us consider the example of measuring $\sigma_z$ (Pauli spin observable) on a qubit. In practice, the observable to be measured ($\sigma_z$ in the present example),  defined by the relative direction of the inhomogeneous magnetic field in the Stern-Gerlach apparatus with respect to the direction of the incoming beam of spin-$1/2$ particles, may not be kept fixed in all experimental runs. Consequently, instead of $\sigma_z$, $\vec{\sigma}.\hat{n}$ will be measured ($\vec{\sigma}.\hat{n}$ = $\sigma_x n_x + \sigma_y n_y + \sigma_z n_z$ with $\sigma_x$, $\sigma_y$, $\sigma_z$ are three Pauli operators and $n_x^2 + n_y^2 + n_z^2 = 1$), where $\hat{n}$ will be different in different runs (all such $\hat{n}$ should be close to $\hat{z}$). Effectively, this will give rise to convex mixing of different measurements. 

Hence, these operations are extremely relevant in practical scenarios. The various layers of incompatibility and their operational implications under these operations are discussed in the subsequent sections.

\section{Classifying measurement incompatibility under coarse-graining and disjoint-convex-mixing}\label{sec4}

%\blk

\subsection{Coarse-graining of outcomes}

%Coarse-graining may occur in practice due to measurement errors. Additionally,  coarse-graining is used in  when one does not require a finer description of the values of the measurement outcomes, but rather a broad description suffices to understand the nature of some physical property.
%Consider a $d$ outcome measurement $\{M_{z}\}$ and an arbitrary coarse-graining that yields a $\tilde{d}$ outcome measurement $\{M_{\tilde{z}}\}$, where $\z\in [\od]  \,\, \text{and} \,\,  z\in [d]$. 

%\re 
Recall that the coarse-graining of a set of $d$-outcome measurements $\{M_{z|x}\}$ produces new measurements defined by
\be \label{gen-cg}
\tilde{M}_{\tilde{z}_x|x} = \sum_{z} c_x(\tilde{z}_x|z) M_{z|x} \,\, ,  \text{ where } \ c_x(\tilde{z}_x|z) \in \{0,1\}.
 \ee 
Before moving forward, let us define what we mean by \textit{trivial coarse-graining}. If the coarse-graining results in a measurement where one of the outcomes, say $\tilde{z}_x$, always occurs, that is, $\tilde{M}_{\tilde{z}_x|x} = \I$, we refer that operation as \textit{trivial coarse-graining} for the measurement labelled by $x$.
The fact that two incompatible measurements may not necessarily remain incompatible after certain coarse-graining motivates the following definition of full incompatibility with respect to (w.r.t.) coarse-graining.
%\blk 

\begin{definition}[Fully incompatible measurements w.r.t. coarse-graining]
A set of measurements $\{M_{z|x}\}$ is fully incompatible w.r.t. coarse-graining if they remain incompatible after all possible nontrivial coarse-graining. That is if the resultant set of measurements given by Eq. \eqref{gen-cg}
after all possible sets of nontrivial coarse-graining
is incompatible, then we call them fully incompatible. Note that, the coarse-graining can be different for different settings $x$.  
\end{definition}  

\begin{definition}[$k$-incompatible measurements w.r.t. coarse-graining]
A set of measurements $\{M_{z|x}\}$ is $k$-incompatible w.r.t. coarse-graining if they remain incompatible after all possible nontrivial coarse-graining that gives rise to at least $k$ outcome measurements. In other words, if the resultant set of measurements given by Eq.~\eqref{gen-cg}, where $\tilde{z}_x\in [d_x]$ and $d_x \geqslant k$ for all $x,$
is incompatible, then we call them $k$-incompatible.  \label{def_2}

\end{definition} 

For example, the three outcome rank-one projective measurement pair, defined by the following (un-normalized) vectors
\be 
M=\big\{|0\rangle,|1\rangle,|2\rangle \big\}
\ee and 
\be 
N=\big\{\ket{0}+|1\rangle ,\ket{0}-|1\rangle ,\ket{2} \big\} 
\ee 
is 3-incompatible, but not 2-incompatible since coarse-graining of the first two outcomes of these measurements yields compatible measurements.

\begin{result} \label{ob1}
A set of fully incompatible measurements is equivalent to 2-incompatible measurements w.r.t. coarse-graining.
\end{result}
\begin{proof}
If a set of measurements is 2-incompatible, then it implies that the set remains incompatible after all possible coarse-graining of the outcomes such that the number of outcomes of each measurement in the newly formed set of measurements is greater than or equal to two. Also, the lowest number of outcomes of measurement is two for a nontrivial coarse-graining. Furthermore, if a set of measurements is $d$-incompatible, then, by definition, it is $n$-incompatible as well, where $n>d$, but the converse is not true. This proves Observation 1.
\end{proof}

\begin{result}\label{ob2}  
Consider two projective measurements, defined by $\{P_i\}, \,\, \text{and} \,\, \{Q_j\}$, where $i\in [d]$ and $j\in [d']$. Let $\{\mathcal{M}_k\}_k$ be the set of all proper subsets of $[d]$, and $\{\mathcal{N}_l\}_l$ be the set of all proper subsets of $[d']$. Then these two measurements are fully incompatible w.r.t. coarse-graining if and only if 
\be \label{c1}
 \left[ \sum_{i \in \mathcal{M}_k} P_i, \sum_{j \in \mathcal{N}_l} Q_j \right] \neq 0, \ \forall k,l .
\ee    
\end{result} 
\begin{proof}
The result is a direct consequence of the fact that for sharp measurement, compatibility and commutativity are equivalent \cite{heinosaari2010nd}. Suppose $\exists \,\, k,l,$ such that the left-hand-side of \eqref{c1} is zero. Then consider the coarse-graining such that the resulting measurements will be $\{\sum_{i \in \mathcal{M}_k} P_i, \I - \sum_{i \in \mathcal{M}_k} P_i\}$ and $\{\sum_{j \in \mathcal{N}_l} Q_j, \I - \sum_{j \in \mathcal{N}_l} Q_j\}$. The resultant measurements will be compatible.  The converse direction holds true from the definition of fully incompatible w.r.t. coarse-graining. 
\end{proof}

\begin{theorem} \label{theo1}
The following condition is necessary but not sufficient for two rank-one projective measurements of dimension $\ge 4$, defined by $\{|\psi_i\ra\}$ and  $\{|\phi_j\ra\}$, to be fully incompatible w.r.t. coarse-graining: 
\be \label{c1psi}
\la \psi_i|\phi_j \ra \neq 0, \ \forall i,j .\ee 
However, the above condition is necessary and sufficient for two 3-dimensional rank-one projective measurements.
\end{theorem}
\begin{proof}
First, note that for sharp measurement, compatibility and commutativity are equivalent \cite{heinosaari2010nd}. Suppose $\exists i,j,$ such that $\la \psi_i|\phi_j\ra =0 $. Consider coarse-graining of all other outcomes except $i$ and $j$ for the two measurements. Then, the resulting measurements will be $\{|\psi_i\ra\!\la \psi_i|, \I - |\psi_i\ra\!\la \psi_i|\}$ and $\{|\phi_j\ra\!\la \phi_j|, \I - |\phi_j\ra\!\la \phi_j|\}$, which commute with each other (i.e., the resulting measurements are compatible) since $\la \psi_i|\phi_j\ra =0 $.  Thus, if the measurements are fully incompatible, condition \eqref{c1psi} holds. 

To show that \eqref{c1psi} is not sufficient, consider the following two $4$-dimensional rank-one projective measurements (with the normalization factor $1/\sqrt{2}$),
\bea \label{ex_cg}
\{\ket{0}+\ket{1}, \ket{0}-\ket{1} , \ket{2}+\ket{3} , \ket{2}-\ket{3} \} \nonumber \\
\{\ket{0}+\ket{2}, \ket{0}-\ket{2} , \ket{1}+\ket{3} , \ket{1}-\ket{3} \} .
\eea 
One can check that \eqref{c1psi} holds for all $i,j = 0,1,2,3$. But coarse-graining of outcomes $0,1$ and $2,3$ for both the measurements leads to 
\bea 
\{|0\ra\!\la 0| +|1\ra\!\la 1|, |2\ra\!\la 2| +|3\ra\!\la 3| \} , \nonumber  \\
\{|0\ra\!\la 0| +|2\ra\!\la 2|, |1\ra\!\la 1| +|3\ra\!\la 3| \},
\eea 
which are compatible.

In 3-dimension, say, the measurements are $M=\{|\psi_i\ra\}$ and $N=\{|\phi_j\ra\}$ with $i,j\in \{1,2,3\}.$ Now any non-trivial coarse-graining yields binary-outcome measurements of the form: $ \{|\psi_i\ra\!\la\psi_i|, \,\, \I-|\psi_i\ra\!\la\psi_i|\} \,\, \text{and} \,\,  \{|\phi_j\ra\!\la\phi_j|, \,\, \I-|\phi_j\ra\!\la\phi_j|\}$. It is easy to see that these two remain incompatible if and only if $[|\psi_i\ra\!\la\psi_i|,|\phi_j\ra\!\la\phi_j|] \neq 0$, which is equivalent to
$\la \psi_i|\phi_j \ra \neq 0,1$. In the case where $\la \psi_i|\phi_j \ra  = 1$, there exists another pair $(i,j')$ such that $\la \psi_{i}|\phi_{j'} \ra  = 0$; thus, \eqref{c1psi} implies fully incompatible in dimension $3$.
\end{proof}

 \begin{theorem} \label{theo2}
Two rank-one projective measurements $M\equiv \left\{|i\rangle \langle i| \right\}_{i=1}^{d}, \,\,  N\equiv \left\{|\psi_{j}\rangle \langle \psi_{j}| \right\}_{i=1}^{d}$ in dimension $d$ are fully incompatible w.r.t. coarse-graining
%if there exists $m, k$ with $m\in S\,\, \text{and} \,\, k\in [d]\backslash S, $ such that the following condition holds:
%\be \label{cGpsi}\sum_{j\in \tilde{S}} \braket{m}{\psi_j}\bra{\psi_j} \ket{k} \neq 0,  \quad \forall   S , \tilde{S} \ee where $S, \tilde{S}\subset [d] \,\, \text{and} \,\, S,\tilde{S} $ is a non-empty set.\end{theorem}
if for any pair of non-empty subsets $S,\tilde{S}\subset [d]$, which corresponds to the coarse-graining of $M$ and $N$, there exist $m, k,$ with $m\in S\,\, \text{and} \,\, k\in [d]\backslash S$, such that 
\be \label{cGpsi}
\sum_{j\in \tilde{S}} \braket{m}{\psi_j}\bra{\psi_j} \ket{k} \neq 0.
\ee 
\end{theorem}
\begin{proof}
It is sufficient to consider all possible binary outcome measurements obtained by coarse-graining of outcomes of $M \,\, \text{and} \,\, N$. Let after a particular coarse-graining $S, \tilde{S}$ the new measurements are $\tilde{M} \equiv \{M_{+}, M_{-}\} \,\, \text{and} \,\, \tilde{N} \equiv \{N_{+}, N_{-}\}$, where $M_+ = \sum_{i \in S} \ketbra{i}{i}, N_{+} = \sum_{j \in \tilde{S}} \ketbra{\psi_j}{\psi_j}, \,\, S,\tilde{S} \subset [d] \,\, \text{and} \,\, S, \tilde{S}$ are non empty set. Now the sufficient condition to prove the theorem is the condition $[M_{+},N_{+}] \ne 0$ for all possible coarse-graining i.e. for all $S, \tilde{S}$.   Now as $\{\ket{m}\bra{k}\} \,\, \text{with} \,\, m,k \in [d]$ forms a basis for a $d\times d$  hermitian matrix, it is sufficient if we find any non-vanishing matrix element of $[M_{+},N_{+}]$, expressed in the basis of $\{\ket{m}\bra{k}\} \,\, \text{for all possible coarse-graining, i.e. for all} \,\, S, \tilde{S}$.   Now we consider any such matrix element for a particular coarse-grainining $S \,\, \text{and} \,\, \tilde{S}$ of the two measurements $M$ and $N$, i.e. $[M_{+},N_{+}]_{(m,k)} = \bra{m} \,\, M_{+}N_{+}-N_{+}M_{+}\,\,\ket{k}, \,\, \text{with}$ 
\bea 
\bra{m} \,\, M_{+}N_{+}\ket{k}&=&\sum_{i\in S}\sum_{j\in \tilde{S}} \braket{i}{\psi_j} \braket{m}{i}\braket{\psi_j}{k} \nonumber \\
&=& \sum_{j\in \tilde{S}} \braket{m}{\psi_j}\braket{\psi_j}{k} \,\, ,
\eea for $m \in S$.  
Similarly, \be \label{NM_matrixElement}
\bra{m} \,\, N_{+}M_{+}\ket{k}=\sum_{i\in S}\sum_{j\in \tilde{S}} \braket{\psi_j}{i} \braket{m}{\psi_j}\braket{i}{k}, 
\ee which we can always set to be zero by choosing $k\notin S \,\, \text{and} \,\, k \in [d], \,\, \text{i.e.} \,\,  k \in [d]\backslash S \,\, \text{for all} \,\, S \,\, \text{and} \,\, \tilde{S}$.  This proves the theorem. 
\end{proof}

%\textbf{Example---} 
We know that the measurements in Eqs.(\ref{ex_cg}) are not fully incompatible with respect to coarse-graining. Hence, according to Theorem \ref{theo2}, these measurements should not satisfy the condition (\ref{cGpsi}). One can indeed check that the measurements (\ref{ex_cg}) do not satisfy the condition (\ref{cGpsi}).
On the other hand, Theorem \ref{theo2}
 can be further refined to show that certain pairs of measurements are fully incompatible.
\begin{coro}\label{coro-1}
Consider two rank-one projective measurements $M\equiv \left\{|i\rangle \langle i| \right\}_{i=1}^{d}, \,\,  N\equiv \left\{|\psi_{j}\rangle \langle \psi_{j}| \right\}_{i=1}^{d}$ in dimension $d$. If there exists $\ket{m}\in M $, such that 
\be \label{cGpsi2}
\sum_{j\in \tilde{S}} \braket{m}{\psi_j}\!\bra{\psi_j} \ket{k} \neq 0.
\ee
holds for all $\ket{k} \in M$ with $k\neq m$, and for all non-trivial subsets $\tilde{S}$ of $[d]$, then the measurements are fully incompatible w.r.t. coarse-graining.
\end{coro}
\begin{proof}
    The condition mentioned above implies the sufficient condition given in Theorem \ref{theo2}. Without loss of generality, we can assume that $m\in S$ for any non-trivial coarse-graining of the first measurement, and there always exists some $k$ such that \eqref{cGpsi} holds, as enforced by the stated condition. 
\end{proof}
 \begin{theorem}
     Consider two mutually unbiased bases in prime dimenions: $\{\ket{i}\}_{i=1}^d$ and $\{\ket{\psi}_j\}_{j=1}^d$, where $\ket{\psi}_j = (1/\sqrt{d}) \sum_{i=1}^d \omega^{ji}\ket{i} $ with $\omega$ being the $d$-th root of unity and $d$ is prime. These bases are fully incompatible w.r.t. coarse-graining.
 \end{theorem}
\begin{proof}
Substituting the expression for $\ket{\psi_j}$ and setting $\ket{m} = \ket{d}$ into the left-hand side of Eq.~\eqref{cGpsi2} yields the sum $\sum_{j \in \tilde{S}} \omega^{j(d-k)}$. To complete the proof as per Corollary \ref{coro-1}, we want to show that this sum is non-zero for all $k \in \{1, \cdots, d-1\}$. 

Since $d$ is prime, the integers modulo $d$ form a field. Consequently, the map $j \mapsto j(d-k)$ permutes the elements of the cyclic group $\mathbb{Z}_d$. The problem, therefore, reduces to showing that $\sum_{j \in \tilde{S}} \omega^{j} \neq 0$ for any non-empty proper subset $\tilde{S} \subset [d]$. This follows from a fundamental property of prime roots of unity: the $d$-th cyclotomic polynomial for a prime $d$ is given by $1 + x + \cdots + x^{d-1}$. This polynomial is the minimal polynomial of $\omega$ over the rationals and is irreducible \cite{Lang70}. Therefore, no non-trivial sum of a proper subset of the $d$-th roots of unity can vanish. This concludes the proof.
\end{proof}
%\blk

\subsection{Disjoint-convex-mixing of measurements}

Consider a device that implements three different measurements, $M=\{M_{z}\}_z,N=\{N_{z}\}_z$, and $R=\{R_{z}\}_z$, all having the same number of outcomes $z\in [d]$. For the first input, it performs measurement $M$ with probability $q$ and $N$ with probability $(1-q)$. For the second input, it always performs measurement $R.$ Thus, the new measurement $Q_{(M,N)}$, realized through the convex-mixing, is
\be \label{QMN}
Q_{(M,N)} = \{q M_{z} + (1-q) N_{z} \}_z, 
\ee
where $q\in [0,1]$ is the weightage of the disjoint-convex-mixing. Even if $R$ is incompatible with $M$ and $N$ separately, $R$ is not necessarily incompatible with $Q_{(M,N)}$ for all values of $q$.
In a similar way, we introduce the notion of full incompatibility w.r.t. disjoint-convex-mixing.

\begin{definition} \label{definition3}
Three measurements $M,N \,\, \text{and} \,\, R$ are fully incompatible w.r.t. disjoint-convex-mixing if each of the pairs, $M$ and $Q_{(N,R)},$ $N$ and $Q_{(M,R)},$ $R$ and $Q_{(M,N)},$  are incompatible for all values of $q$, where $Q_{(\ ,\ )}$ is defined in \eqref{QMN}.
\end{definition} 

Consider three unbiased qubit binary outcome measurements $\{M_0,M_1\}, \{N_0,N_1\}, \{R_0,R_1\}$, 
\bea \label{qum}
&M_{z} = \frac{1}{2} \left( \I + (-1)^z \,\,\Vec{n_0} \cdot \Vec{\sigma} \right), \nonumber \\
&N_{z}= \frac{1}{2} \left( \I + (-1)^z \,\, \Vec{n_1} \cdot \Vec{\sigma} \right), \nonumber \\
&R_{z} =\frac{1}{2} \left( \I + (-1)^z \,\, \Vec{n_2} \cdot \Vec{\sigma} \right),
\eea with $z=0,1$  
and $||\Vec{n_i}|| \leqslant 1 \,\, \text{where} \,\, i\in \{0,1,2\}$.  A necessary and sufficient criterion for the incompatibility of two unbiased binary-outcome qubit measurements is given by  Busch in \cite{busch86} (also, see eq.(7) of \cite{guhne2023colloquium}).  By applying this criterion, we find that the above three measurements \eqref{qum} are fully incompatible w.r.t. disjoint-convex-mixing if and only if,
\be \label{c2}
|| \Vec{n}_i + q \Vec{n}_{j} + (1-q) \Vec{n}_{k} || + || \Vec{n}_i - q \Vec{n}_{j} - (1-q) \Vec{n}_{k} || > 2 , 
\ee 
for all $q$, and for all $(i,j,k)\in \{(0,1,2),(1,2,0),(2,0,1)\}$. 

% \begin{theorem}\label{r3}
% If three-qubit measurements \eqref{qum} are such that $\Vec{n}_i$ are in the same plane of the Bloch sphere, then they are not fully incompatible w.r.t. disjoint-convex-mixing.
% \end{theorem} 
% \begin{proof}
% If the three $\Vec{n}_i$ are in the same plane, then there exists at least one triple $(i,j,k)$ such that 
% \be 
% \Vec{n}_0 = \frac1c (q \Vec{n}_{1} + (1-q) \Vec{n}_{2})
% \ee 
% for some $ c \in (0,1],q \in [0,1]$. In other words, there exists a triple $(i,j,k)$ so that $\Vec{n}_i$ is expressed as a linear combination of $\Vec{n}_{j}$ and $\Vec{n}_{k}$ with non-negative coefficients (here $q/c$ and $(1-q)/c$), where the sum of those two non-negative coefficients is greater than or equal to 1.
% Substituting this into left hand side of \eqref{c2}, we find 
% \be 
% || (1+c)\Vec{n}_i || + || (1-c) \Vec{n}_i || \leqslant |1+c|+|1-c| = 2 .
% \ee 
% This contradicts with \eqref{c2}, implying they are compatible.

% \end{proof}

\begin{theorem}\label{r3}
If three-qubit measurements \eqref{qum} are such that $\Vec{n}_i$'s are in the same plane of the Bloch sphere, then they are not fully incompatible w.r.t. disjoint-convex-mixing.

\end{theorem} 

\begin{proof}
If the three $\Vec{n}_i$'s are in the same plane, then there exists at least one triple $(i,j,k)$ from any set $\{i,j,k\}$ such that 
\bea 
q \Vec{n}_{j} + (1-q) \Vec{n}_{k} & \propto & \Vec{n}_{i} \nonumber \\
&=& c \,\, \Vec{n}_{i},            
\eea
where $ q \in [0,1] \,\, \text{and} \,\, c$ is the proportionality constant, $ c\in \Bbb{R}, $ satisfying $|| c\,\,\Vec{n}_i ||\le 1.$ Now, as $\frac{1}{2} \left( \I  \pm \,\,\Vec{n_i} \cdot \Vec{\sigma} \right)$ and $\frac{1}{2} \left( \I \pm \,\, c \,\,\Vec{n_i} \cdot \Vec{\sigma} \right)$ are always compatible which is very straight-forward, also it follows from the compatibility condition given by Busch in \cite{busch86}, i.e.,  Eq.(7) of \cite{guhne2023colloquium}; this proves the theorem.
\end{proof}

\begin{theorem}\label{r4}
Consider three-qubit measurements \eqref{qum} are such that $\Vec{n}_0 = \nu_0 \hat{x}$, $\Vec{n}_1 = \nu_1 \hat{y}$, $\Vec{n}_2 = \nu_2 \hat{z}$ with $0 \leqslant \nu_0 ,\nu_1, \nu_2\leqslant 1$, that is, the noisy version of Pauli observables,
\bea \label{noisy_pauli}
&M_z = \frac{1}{2} \left( \I + (-1)^z \nu_0 \sigma_x \right) = \nu_0 \left(\frac{ \I + (-1)^z \sigma_x }{2}\right) + (1-\nu_0 )\frac{\I}{2} , \nonumber \\ 
&N_z = \frac{1}{2} \left( \I + (-1)^z \nu_1 \sigma_y \right) = \nu_1 \left(\frac{ \I + (-1)^z \sigma_y }{2}\right) + (1-\nu_1 )\frac{\I}{2} , \nonumber \\
&R_z = \frac{1}{2} \left( \I + (-1)^z \nu_2 \sigma_z \right) = \nu_2 \left(\frac{ \I + (-1)^z \sigma_z }{2}\right) + (1-\nu_2 )\frac{\I}{2},\nonumber \\
\eea 
with $z=0,1$. These measurements are fully incompatible w.r.t. disjoint-convex-mixing if and only if 
\be   \label{c3}
\min  \bigg\{ \nu^2_0 + \frac{\nu^2_1 \nu^2_2}{\nu^2_1 + \nu^2_2}, \nu^2_1 +\frac{\nu^2_0 \nu^2_2}{\nu^2_0 + \nu^2_2},  \nu^2_2 + \frac{\nu^2_0 \nu^2_1}{\nu^2_0 + \nu^2_1} \bigg\} > 1.
\ee   
\end{theorem}
\begin{proof}
In terms of $\nu_i$, the left hand side of \eqref{c2} becomes
%\be \label{xyzcm}
$2 \sqrt{\nu^2_i + q^2 \nu^2_j+ (1-q)^2 \nu^2_k }$.
%\ee 
Note that the minimum of $q^2 \nu^2_j+ (1-q)^2 \nu^2_k$ occurs at $\tilde{q} =\nu_k^2/(\nu^2_j + \nu^2_k)$. 
Since $0 \leq \nu_k^2/(\nu^2_j + \nu^2_k) \leq 1$ for any $\nu_j, \nu_k \in [0,1]$, the above-mentioned minimum can always be achieved with a suitable choice of $q$. Hence, we have that
\begin{align}
2 \sqrt{\nu^2_i + q^2 \nu^2_j+ (1-q)^2 \nu^2_k } &\geq 2 \sqrt{\nu^2_i + \tilde{q}^2 \nu^2_j+ (1-\tilde{q})^2 \nu^2_k } \nonumber \\
& = 2 \sqrt{\nu^2_i + \frac{ \nu^2_j \nu^2_k}{\nu^2_j + \nu^2_k} }
\end{align}
Thus, the right-land-side of \eqref{c2} is greater than 2 for all values of $q$ if and only if
\be 
\nu^2_i + \frac{ \nu^2_j \nu^2_k}{\nu^2_j + \nu^2_k} > 1 . 
\ee 
Taking all the three possible values of $(i,j,k)$ we get the condition \eqref{c3}.
\end{proof}
Clearly, the three Pauli observables are fully incompatible w.r.t. disjoint-convex-mixing, and moreover, if we take an equal amount of noise $\nu_0=\nu_1=\nu_2= \nu$, then \eqref{c3} implies $\nu > \sqrt{2/3}$. It is worth mentioning that the three Pauli observables are incompatible in the usual sense (given by Eq. (\ref{def_comp}))  for $\nu > 1/\sqrt{3}$, which was shown in \cite{Heinosaari2008}. Thus, for $1/\sqrt{3} < \nu \le \sqrt{2/3}$, the three measurements are not fully incompatible w.r.t. disjoint-convex-mixing, though they are incompatible. \\

%%%%%%%%%%%%%%%%%Due3%%%%%%%%%%%%%%%%

We can generalize the notion of full incompatibility w.r.t. disjoint-convex-mixing.

\begin{definition}[$k$-incompatible measurements w.r.t. disjoint-convex-mixing] \label{def_4}
Given a set of $n$ measurements, the measurements are $k$-incompatible w.r.t. disjoint-convex-mixing if, after taking every possible disjoint-convex-mixing that yields $k$ number of measurements, the resulting measurements are incompatible. 
\end{definition}

\begin{definition}[Fully incompatible measurements w.r.t. disjoint-convex-mixing]
A set of $n$ measurements is fully incompatible w.r.t. disjoint-convex-mixing if it is $k$-incompatible for all $k = 2,\cdots,n$. 
\end{definition}

\begin{result}
Fully incompatible measurements w.r.t. disjoint-convex-mixing imply that every pair of measurements from that set is incompatible. The reverse implication does not hold.
\end{result}
\begin{proof}
Consider a set of $n$ measurements, in which there is a pair of measurements $\{M_z\}$ and $\{N_z\}$ that are compatible. Now if we make two partitions where $\{M_z\}$ and $\{N_z\}$ are in different partitions and the disjoint-convex-mixing is such that the probabilities of arising all other measurements are zero, then the resultant measurement pair must be compatible. This is true for any compatible pair of measurements. Thus, if the measurements are fully incompatible w.r.t. disjoint-convex-mixing, every pair of measurements must necessarily be incompatible.

The converse is not true. Consider the three noisy Pauli measurements of Eq.\eqref{noisy_pauli}. It can be shown by using semi-definite programming that if $0.71<\nu\leqslant  0.81$, the measurements are pairwise incompatible, but it is not fully incompatible w.r.t. disjoint-convex-mixing \cite{prl19_diIncom}. 
\end{proof}

Some explicit examples of $k$-incompatible measurements w.r.t. both of these operations are provided in the subsequent section. 

\section{Robustness under noise for different levels of incompatibility} \label{sec5}
In this section, we analyze the role of noise on quantum measurements and study how the incompatibility properties depend on it.  More specifically, we study the variation of the critical amount of noise in the different layers of incompatibility under these classical operations (coarse-graining and disjoint-convex-mixing).  Due to the ubiquitous nature of noise, it is pertinent to study the extent to which noise could
be tolerated by a set of measurements while still retaining their incompatibility. We take a noisy version of mutually unbiased bases measurements (MUBs), 
\be 
M_{i|x} =  \nu\ket{\phi_{i|x}}\bra{\phi_{i|x}}+\frac{(1-\nu)}{d}\I_{d} ,
\ee
where $\{\ket{\phi_{i|x}}\}_{i,x}$ 
form mutually unbiased bases measurements in $\mathbb{C}^d$. Here $\nu$ is the robustness parameter (or visibility parameter) and $(1-\nu)$ is the noise parameter, $0\le \nu \le 1$. When the noise parameter is zero (i.e., robustness, $\nu=1$), the measurements are fully incompatible, and when the noise parameter is one (i.e., the robustness, $\nu=0$), the measurements are trivial and compatible.  Our aim is to obtain the critical value of the robustness parameter above which the measurements remain incompatible before and after different classical operations. 

To check the compatibility, i.e., the existence of a parent POVM, we use the method described in \cite{guhne2023colloquium}. This can be cast as a semi-definite programming (SDP) problem that takes a set of measurements $\{M_{z_x|x}\}$ and deterministic classical post-processings $p(z_x|x,\lambda)$ as input, and checks whether the measurements are compatible or not, subject to the constraints
\bea 
 & \sum_{\lambda} p(z_x|x,\lambda)G_{\lambda}=M_{z_x|x} \forall x, z_x,\\
 &    \sum_{\lambda}G_{\lambda}=\I, \\
 &  G_{\lambda} \ge \mu \I \label{positivity},
\eea
where $\mu$ is the optimization parameter. This method finds the maximum value of $\mu$ for each $\{p(z_x|x,\lambda)\}$. If this optimization returns a negative value of $\mu$, then the constraint of Eq.\eqref{positivity} cannot be fulfilled, which implies that the measurements $\{M_{z_x|x}\}$ are incompatible. Otherwise, they are compatible. 

 As discussed earlier, the degree of incompatibility of measurements may be reduced due to coarse-graining of outcomes and convex mixing of measurements. In the subsequent subsection, we present an explicit analysis of how the degree of incompatibility (quantified by the robustness parameter \citep{heinosaari2015noise}) varies as the measurements are subjected to coarse-graining and disjoint-convex-mixing, in dimensions $3$ and $4$. For simplicity, we have considered the measurements to be MUBs (as discussed earlier), where the number of outcomes are same as that of their dimension. Also, as MUBs exhibit the highest degree of incompatibility, \citep{designolle2019quantifying} it is a natural choice to use them for the prominent observation of the various layers of incompatibility under the classical operations. However, one can apply this formalism to any set of measurements in a similar fashion. 
 % status of their incompatibility with respect to coarse-graining and disjoint-convex-mixing is a natural shoot-off of the present analysis.Since any non-trivial coarse-graining requires the initial measurements (measurements before coarse-graining) to have at least three outcomes, so the dimension of the measurements should be greater than or equal to $3$ in general (here we have considered the measurements in dimension $3$ and $4$) for the case of coarse-graining. Without any loss of generality, we have considered the same dimensions for the study of disjoint-convex-mixing of measurements.
 \subsection{Robustness of incompatibility under coarse-graining of outcomes}
%\blk
\textbf{Dimension $3$.}  Let $\{M_i\}$ and $\{N_j\}, i,j \in \{1, 2, 3\}$ be two three-outcome measurements acting on $\mathbb{C}^3$, where 
\begin{align}
   &M_i=\nu\ket{i}\bra{i}+(1-\nu)\frac{\I}{3}, \nonumber\\ 
   &N_j=\nu\ket{\psi_j}\bra{\psi_j}+(1-\nu)\frac{\I}{3} \label{measurement1},
   \end{align}
   with
   \begin{align}
     &\ket{\psi_0}=\frac{1}{\sqrt{3}}(\ket{0}+\ket{1}+\ket{2}), \nonumber \\
&\ket{\psi_1}=\frac{1}{\sqrt{3}}(\ket{0}+\omega\ket{1}+\omega^2\ket{2}), \nonumber \\  &\ket{\psi_2}=\frac{1}{\sqrt{3}}(\ket{0}+\omega^2\ket{1}+\omega\ket{2})  \label{measurement},
\end{align}
 $\omega$ being the cube roots of unity.  It can be easily checked by SDP that the measurements in Eqs.\eqref{measurement1} and \eqref{measurement} are incompatible (or equivalently, $3$-incompatible using the definition \ref{def_2} ) for robustness parameter  $(\nu) > 0.683$. Now, any non-trivial coarse-graining reduces the measurement outcomes to two, so one can check the incompatibility after all possible choices of such non-trivial coarse-grainings. The measurements are $2$-incompatible (by the definition \ref{def_2}) for $\nu > 0.711$. For $0.683 < \nu \le 0.711$, the measurements are $3$-incompatible, but not $2$-incompatible. This represents the case where the initial measurements, although taken to be incompatible, become compatible after a non-trivial coarse-graining. However, for $\nu > 0.711$, the measurements are $2$-incompatible (and hence $3$-incompatible), suggesting that they remain incompatible after all non-trivial coarse-grainings. These findings suggest that $2$-incompatible measurements exhibit a stronger form of incompatibility compared to the $3$-incompatible measurements (hence supporting Observation \ref{ob1}), and hence are potential candidates for any resource-theoretic applications that harness incompatibility. Note that for three outcome measurements, any non-trivial coarse-graining will reduce the outcomes to two. So, the number of outcomes is always $\le 3$ (where $3$ corresponds to the case where there is no coarse-graining). So, $4$-incompatibility is not applicable (N.A.) here.
 
\textbf{Dimension $4$.} For checking incompatibility w.r.t coarse-graining in $\mathbb{C}^4$, the same procedure is repeated taking two POVM measurements with four outcomes each. The corresponding measurements are
$\{M'_i\}, \{N'_j\} , i,j \in \{1,2,3,4\}$:
\bea\label{measurement2}
M'_i &=& \nu \ket{i}\bra{i}+(1-\nu)\frac{\I}{4}, \nonumber \\
N'_j &=& \nu \ket{\psi'_j}\bra{\psi'_j}+(1-\nu)\frac{\I}{4}, 
\eea
with
\bea\label{description2}
\ket{\psi'_0}&=& \frac{1}{2}(\ket{0}+\ket{1}+\ket{2}+\ket{3}), \nonumber\\
\ket{\psi'_1}&=& \frac{1}{2}(\ket{0}+\ket{1}-\ket{2}-\ket{3}), \nonumber\\
\ket{\psi'_2}&=& \frac{1}{2}(\ket{0}-\ket{1}-\ket{2}+\ket{3}), \nonumber\\
\ket{\psi'_3}&=& \frac{1}{2}(\ket{0}-\ket{1}+\ket{2}-\ket{3}).
\eea

 These measurements in Eqs.\eqref{measurement2} and \eqref{description2} are incompatible (or equivalently $4$-incompatible, using the definition \ref{def_2}) for $\nu > 0.666$. In this case, any non-trivial coarse-graining will reduce the measurement outcomes to three or two, and hence we shall have $3$-incompatiblity and $2$-incompatiblity respectively. For $0.666 < \nu \le 0.675$, the incompatible measurements become compatible after all possible coarse-grainings that reduce the outcomes to at least three. Further reducing the number of outcomes by coarse-graining gives a tighter bound for the critical value of the robustness parameter above which the measurements are incompatible. Hence, for $0.675 < \nu \le 0.720$, the measurements are $3$-incompatible but become compatible if one further reduces the number of outcomes by coarse-graining. These measurements exhibit the strongest form of incompatibility for $\nu > 0.720$ (as also discussed in Observation \ref{ob1}), i.e., they remain incompatible after all possible non-trivial coarse-grainings. The results for coarse-graining of measurement outcomes in dimensions $3$ and $4$ are summarized in TABLE \ref{tab:table-name1}.
%\re In this case, before coarse-graining the measuremen in (\ref{measurement2}) remain incompatible for robustness, $\nu > 0.666.$ After coarse-graining, the measurements are 3-incompatible w.r.t coarse-graining for robustness, $\nu>0.675$ and fully-incompatible (2-incompatible) w.r.t coarse-graining for $\nu>0.720.$  The results are summarized in TABLE \ref{tab:table-name1}.
%\blk

 \begin{table}[h!] 
\begin{tabular}{ c| c| c }
 & \textbf{In dimension $3$} & \textbf{In dimension $4$}  \\
 \hline 
 \textbf{$4$-incompatible}  & N.A. & 0.666 \\  
 \hline  
 \textbf{$3$-incompatible} & 0.683 & 0.675 \\
  \hline
 \textbf{$2$-incompatible} & 0.711 & 0.720 \\
 \hline
\end{tabular}
\caption{\label{tab:table-name1}Critical values of robustness for MUBs w.r.t. coarse-graining in  $\mathbb{C}^3$ and $\mathbb{C}^{4}$, above which the measurements are incompatible. }
\end{table}  

\subsection{Robustness of incompatibility under disjoint-convex-mixing of measurements} 

\textbf{Dimension $3$.} For checking incompatibility w.r.t. disjoint-convex-mixing (CM), a minimum of three measurements are needed. The two measurements are $\{M_i\}$ and $\{N_j\}$, as given in Eqs.\eqref{measurement1} and Eq.\eqref{measurement}. The third measurement is  $\{R_{k}\}, k\in \{1,2,3\} $ where
\begin{equation}
R_k=\nu\ket{\phi_k}\bra{\phi_k}+(1-\nu)\frac{\I}{3}.\label{measurement3}
\end{equation}
with 
\begin{align}
     &\ket{\phi_0}=\frac{1}{\sqrt{3}}(\omega\ket{0}+\ket{1}+\ket{2}), \nonumber \\
&\ket{\phi_1}=\frac{1}{\sqrt{3}}(\ket{0}+\omega\ket{1}+\ket{2}), \nonumber \\  &\ket{\phi_2}=\frac{1}{\sqrt{3}}(\ket{0}+\ket{1}+\omega\ket{2})  \label{phi},
\end{align}

% \begin{equation}\label{phi}
% \begin{split}
%     & \ket{\phi_0}=\frac{1}{\sqrt{3}}(\omega\ket{0}+\ket{1}+\ket{2}), \,\,  \ket{\phi_1}=\frac{1}{\sqrt{3}}(\ket{0}+\omega\ket{1}+\ket{2}),\,\, &\ket{\phi_2}=\frac{1}{\sqrt{3}}(\ket{0}+\ket{1}+\omega\ket{2}}).
%     \end{split}  
%     \end{equation}
     These three measurements are incompatible (equivalently $3$-incompatible, by the definition \ref{def_4}) for $\nu > 0.537$. From these three incompatible measurements, one can effectively reduce the number of measurements to two, by taking a new measurement as the disjoint-convex-mixing of any two measurements from the initial set of three measurements (this can be done in $3$ ways), while keeping the third measurement same. This process of disjoint-convex-mixing may reduce the degree of incompatibility, as discussed before and evidenced for $0.537 < \nu \le 0.764$. This corresponds to the case where the two measurements become compatible after considering all possible disjoint-convex-mixing across all partitions. For $\nu > 0.764$, the measurements are incompatible, suggesting that this is the critical value above which, the measurements are robust against disjoint-convex-mixing. 
%\re Before convex mixing the three measurements remain incompatible for robustness, $\nu>0.537$, and 2-incompatible w.r.t convex mixing for robustness, $\nu>0.764.$     
%The results are summarized in  TABLE \ref{tab:table-name2}. \blk

% These measurements are triplewise incompatible for $\nu >0.537$. For checking  $2$-incompatibility, we define a new measurement  by taking the disjoint-convex-mixing of any of the two measurements with convex coefficient $q$ and the incompatiblity of this new measurement with the remaining $3$rd measurement is checked.
% The value of $\nu$ above which the measurements are incompatible for all values of $q$, taking all possible measurement combinations and all possible permutations of the outcomes of the two measurements gives the value of $2$-incompatible. These measurements are fully incompatible for $\nu > 0.764$. Note that for $0.537<\nu<0.764$, the incompatible measurements become compatible by taking certain disjoint-convex-mixing. 

\textbf{Dimension $4$.} Consider the measurements $\{M'_{i}\},\{N'_{j}\} $ as defined in Eqs. \eqref{measurement2} and \eqref{description2} and another measurement $\{R'_{k}\}$ where $ i,j,k \in \{1,2,3,4\}$,
\begin{equation}
    R'_k =\nu\ket{\phi'_k}\bra{\phi'_k}+(1-\nu)\frac{\I}{4}  \label{phi'}
    \end{equation}
with 
\bea
 \ket{\phi'_0} &=& \frac{1}{2}(\ket{0}-\ket{1}-\mathbbm{i}\ket{2} - \mathbbm{i} \ket{3}),\nonumber \\
\ket{\phi'_1} &=& \frac{1}{2}(\ket{0}-\ket{1}+\mathbbm{i}\ket{2}+\mathbbm{i}\ket{3}),\nonumber \\
\ket{\phi'_2} &=& \frac{1}{2}(\ket{0}+\ket{1}+\mathbbm{i}\ket{2}-\mathbbm{i}\ket{3}),\nonumber \\
\ket{\phi'_3} &=& \frac{1}{2}(\ket{0}+\ket{1}-\mathbbm{i}\ket{2}+\mathbbm{i}\ket{3}).
\eea

 These three measurements in dimension $4$ are incompatible for $\nu > 0.692$. As discussed before, disjoint-convex-mixing makes them compatible and they continue to remain so till $\nu = 0.705$. So, for $0.692 < \nu \le 0.705$, the three measurements are incompatible, but disjoint-convex-mixing of any two of them makes them compatible. For $\nu > 0.705$, the measurements are incompatible and robust against disjoint-convex-mixing across any partition that effectively brings down the number of measurements to two. The results for disjoint-convex-mixing of measurements in dimensions $3$ and $4$ are summarized in TABLE \ref{tab:table-name2}. 

\begin{table}[h!]
\begin{tabular}{ c| c| c }
 & \textbf{In dimension $3$} & \textbf{In dimension $4$}  \\
 \hline 
 \textbf{$3$-incompatible}  & 0.537 & 0.692 \\  
  \hline
 \textbf{$2$-incompatible} & 0.764 & 0.705 \\
 \hline
\end{tabular}
\caption{\label{tab:table-name2} Critical values of robustness for MUBs w.r.t. disjoint-convex-mixing in $\mathbb{C}^3$ and $\mathbb{C}^4$, above which the measurements are incompatible.}
\end{table}
\subsection{Robustness of incompatibility under disjoint-convex-mixing of measurements followed by coarse-graining of outcomes } In the previous sub-sections, we analysed the variation of the robustness parameter for different layers of incompatibility under elementary classical operations performed either on the inputs, or on the outputs.  Below, we investigate this direction when the two operations are performed together.

\textbf{Dimension $3$.} Consider the three measurements given by Eqs.\eqref{measurement1}, and \eqref{measurement3}, with the corresponding projectors given by Eqs.\eqref{measurement} and \eqref{phi} respectively. Now, all possible convex mixing of any two measurements makes them $2$-incompatible, and hence fully incompatible (according to definition \ref{def_4}) under disjoint-convex-mixing  for $\nu > 0.764$ (as observed from TABLE \ref{tab:table-name2}). This is followed by the coarse-graining of the outcomes of these two measurements. All possible coarse-graining of the outcomes results in the measurements being $2$-incompatible (according to definition \ref{def_2}) for $\nu > 0.894$. So, for $0.764 < \nu \le 0.894$, post-processing reduces the degree of incompatibility of the initial pre-processed incompatible measurements. The results are summarized in TABLE \ref{tab:table-name3}.
\begin{table}[h!]
\begin{tabular}{ c| c|}
$\mathbf{d=3}$  & \textbf{DCM followed by CG}  \\
 \hline 
 \textbf{$3$-incompatible}   & 0.764 \\  
  \hline
 \textbf{$2$-incompatible}  & 0.894 \\
 \hline
\end{tabular}
\caption{\label{tab:table-name3} Critical values of robustness for MUBs w.r.t. disjoint-convex-mixing (DCM) of inputs followed by coarse-graining (CG) of outcomes in $\mathbb{C}^3$, above which the measurements are incompatible.}
\end{table}

\textbf{Dimension $4$.} The three measurements given by Eqs.\eqref{measurement2}, \eqref{description2}, and \eqref{phi'} are incompatible for $\nu > 0.692$. These measurements become compatible after all possible disjoint-convex-mixing resulting in two measurements for $\nu \le 0.705$, as shown in TABLE \ref{tab:table-name2}. Post-processing of the outcomes further reduces the degree of incompatibility. These measurements are $3$-incompatible, and $2$-incompatible for $\nu > 0.774$ and $\nu > 0.813$ respectively under all possible coarse-graining of the outcomes. For $0.705 < \nu \le 0.774$, the measurements become compatible by relabelling any one outcome, i.e., they are $4$-incompatible, but not $3$-incompatible (following definition \ref{def_2}). Further post-processing increases the compatibility bound to $0.813$. Note that for an initial set of three measurements, the concept of $4$-incompatiblity under disjoint-convex-mixing is not applicable (N.A.). The results are summarized in TABLE \ref{tab:table-name4}.

These examples highlight the importance of considering classical operations on both inputs and outputs. The degree of incompatibility gets reduced by disjoint-convex-mixing and subsequently decreases further by coarse-graining of outcomes. Hence, the coarse-grained $2$-incompatible measurements exhibit the strongest form of incompatibility when the elementary classical operations are considered on both inputs and outputs. 
\begin{table}[h!]
\begin{tabular}{ c| c|}
$\mathbf{d=4}$  & \textbf{DCM followed by CG}  \\
 \hline 
 \textbf{$4$-incompatible}  & 0.705 \\  
 \hline
 \textbf{$3$-incompatible}   & 0.774 \\  
  \hline
 \textbf{$2$-incompatible}  & 0.813 \\
 \hline
\end{tabular}
\caption{\label{tab:table-name4} Critical values of robustness for MUBs w.r.t. disjoint-convex-mixing (DCM) of inputs followed by coarse-graining (CG) of outcomes in $\mathbb{C}^4$, above which the measurements are incompatible.}
\end{table}

%\blk
    %\end{center}
  
%\end{widetext}
%

%%%%%%%%%%%%%%%%%%%%%%%%%%%%%%%%%%%%%%%%%%%%

\section{Operational witnesses of incompatibility} \label{sec6}

In this section, we explore how different levels of incompatibility can be witnessed through the input-output statistics derived from characterized devices.

\subsection{Coarse-graining of outcomes} 

As we mentioned earlier, measurements that are fully incompatible w.r.t. coarse-graining show stronger incompatibility compared to the measurements that are not fully incompatible w.r.t. coarse-graining. To operationally certify whether a set of measurements is fully incompatible w.r.t. coarse-graining is important from the perspective of determining the practical utility of
such a set for revealing phenomena such as Bell inequality, steering and
contextuality \cite{Fine1982,PRL103230402,PhysRevLett.113.160403,PhysRevLett.113.160402,tp1,PhysRevLett.115.230402,tp2,LIANG20111,Xu2019,sg_prl}, as well as for checking the proficiency of such a set for
information processing tasks \cite{Carmeli2020,SahaIncom}. Certification means that we need to infer the incompatibility of the measurements from the input-output statistics of the measurement device without knowing its internal functioning. Below we
study the two classes of certification separately.

%\textit{Device-independent witness.---}
\subsubsection*{Device-independent witness }
Violation of Bell inequality provides device-independent witness for incompatible measurements. Here, one neither requires any prior knowledge of the internal functioning of the measurement device nor any idea of the dimension of the system on which the measurements act.  However, not all incompatible measurements yield violations of Bell inequalities \cite{Bene_2018,PhysRevA.97.012129}. Nonetheless, we have the following results.

\begin{result}
    Full-incompatibility w.r.t. coarse-graining of any two measurements can always be witnessed in a device-independent way.
\end{result}
\begin{proof}
In \cite{PRL103230402}, it has been proven that any pair of binary-outcome incompatible measurements violate at least one Bell-CHSH inequality by suitably choosing a shared entangled state between the parties and suitably choosing measurements on the other subsystem. On the other hand, from Observation \ref{ob1}, we know if two measurements are full-incompatible w.r.t. coarse-graining, then they must be two-incompatible w.r.t. coarse-graining.  Combining these two facts, we can conclude this observation. 
\end{proof}

\subsubsection*{Device-independent witness from a single Bell experiment}

In the above-mentioned approach, different Bell experiments are probed for different coarse-graining. It is interesting to explore whether different levels of incompatibility w.r.t. coarse-graining can be witnessed universally through a single Bell experiment in a device-independent manner.
%Here "single experiment" implies the Bell test which gives a Bell-inequality violation for the two measurements of outcomes $d_1$ and $d_2$ respectively as chosen as inputs of Alice and some measurement choices for Bob and a shared entangled state. As we know if two measurements are $2-$ incompatible w.r.t. coarse-graining then they are fully-incompatible w.r.t. coarse-graining. Also, we use the result of \cite{PRL103230402}, that any two binary outcome incompatible measurements violate the Bell-CHSH inequality. 
In this case, the objective is to consider different coarse-grained statistics of a Bell experiment and check whether those statistics have local explanations or not. 
%If, all possible coarse-grained two-outcome measurements violate the Bell-CHSH inequalities then the measurements are fully-incompatible w.r.t. coarse-graining.   

To understand how this technique works, we consider a well-known example of two three-outcome rank-one projective measurements in $\mathbb{C}^{3}$ which give maximum violation of the Collins-Gisin-Linden-Massar-Popescu (CGLMP) inequality \cite{cglmp}, a well-studied 
\cite{prl.100.120406, physreva.99.022305, jphysa.55.384011}  Bell-type inequality.  The projective  measurements of Alice are: $A_a\equiv\{|\xi\ra_{A,a}\},$ where \be \label{cglm_measurement}
|\xi\ra_{A,a} =\frac{1}{\sqrt{3}} \sum_{j=0}^2 \text{exp}\bigg(\mathbbm{i}\frac{2\pi}{3}j(\xi+\alpha_{a})\bigg)|j\ra_{A},
\ee with $a\in \{1,2\}$ corresponding to two different measurement settings of Alice, and $\xi \in \{0,1,2\},\alpha_1=0, \alpha_2=\frac{1}{2}$. The projective measurements of Bob are:  $B_b\equiv\{|\eta\ra_{B,b}\},$ where \be 
|\eta\ra_{B,b} =\frac{1}{\sqrt{3}} \sum_{j=0}^2 \text{exp}\bigg(\mathbbm{i}\frac{2\pi}{3}j(-\eta+\beta_{b})\bigg)|j\ra_{B}, 
\ee with, $b \in \{1,2\}$ corresponding to two different measurement settings of Bob, and $\eta\in \{0,1,2\}, \beta_1=\frac{1}{4}, \beta_2=-\frac{1}{4}.$ 

 Specifically, here we investigate the different layers of incompatibility of the measurements given in (\ref{cglm_measurement}) under coarse-graining. 
The incompatibility of the measurements of Eq.\eqref{cglm_measurement} before coarse-graining is guaranteed by the CGLMP-inequality violation. Now we give a theoretical analysis to investigate the incompatibility status of the measurements of Alice given in \eqref{cglm_measurement} after all possible coarse-grainings. The results are summarized in TABLE \,\,\ref{table3}.
 
 \textit{Theoretical explanation.--- }Out of three outcomes if any of the two outcomes are coarse-grained it becomes a (2,2,2,3) scenario, i.e., two inputs for Alice and two inputs for Bob and each input of Alice has two outcomes and each input of Bob has three outcomes. 
   For this scenario, Collins and Gisin have shown that there are a total of 72 CH-facet inequalities \cite{cg_04}. The CH-inequalities \cite{fine_prl} are of the form: 

%\begin{widetext}

\be \label{ch_facet}
-1\le S \le 1 ,
\ee
with CH-functional \bea \label{ch_functional}
S= &P(00|A_{1},B_{1})+P(00|A_{1},B_{2})+P(00|A_{2},B_{2})  \nonumber 
\\
    &-P(00|A_{2},B_{1})-P(0|A_1)-P(0|B_2).
\eea

% \bea \label{ch_facet}
%     -1 \le & P(00|A_{1},B_{1})+P(00|A_{1},B_{2})+P(00|A_{2},B_{2})  \nonumber \\
%     &-P(00|A_{2},B_{1})-P(0|A_1)-P(0|B_2) \le 0. 
%     \eea
%\end{widetext} 
The other CH-inequalities are obtained by (1) interchanging $A_1$ with $A_2$, (2) interchanging $B_1$ with $B_2$, and (3) interchanging both $A_1$ with $A_2$ and $B_1$ with $B_2$. Consider the scenario where we coarse-grain (0,1) outcomes for both the inputs of Alice. Lets make the following relabeling $(0,1)\equiv \overline{0}$ and $ 2 \equiv \overline{1} $ for the outcomes of $A_1$ and $A_2$ and also consider the clubbing of $(0,1) \equiv \overline{0}$ outcomes both for $B_1$ and $B_2.$ Under this relabelling Eq.(\ref{ch_facet}) takes the form: %\begin{widetext}
    \bea\label{ch_1}
    -1 \le & P(\overline{0}\,\, \overline{0}|A_{1},B_{1})+P(\overline{0}\,\,\overline{0}|A_{1},B_{2})+P(\overline{0}\,\,\overline{0}|A_{2},B_{2}) \nonumber \\
    &-P(\overline{0}\,\,\overline{0}|A_{2},B_{1})-P(\overline{0}|A_1)-P(\overline{0}|B_2) \le 0, \eea
%\end{widetext}

%\begin{widetext}
  where, \bea
  P(\overline{0}\,\, \overline{0}|A_{i},B_{j}) = P(00|A_{i},B_{j})+P(01|A_{i},B_{j})\nonumber\\
  +P(10|A_{i},B_{j})+P(11|A_{i},B_{j}), 
  \eea
  \bea
&P(\overline{0}|A_{i})=P(0|A_{i})+P(1|A_{i}),\nonumber \\
&P(\overline{0}|B_{j})=P(0|B_{j})+P(1|B_{j}), \,\, i,j \in \{1,2\}.
\eea

   Similarly, there are eight other possible clubbings for Bob's measurement outcomes and for each clubbing, we have facet inequalities similar to Eq.(\ref{ch_1}). 
   
   One can check that when there are the same coarse-graining of outcomes for both measurement inputs of Alice, we get a CH-inequality violation. Thus, under these coarse-grainings, the two measurements of Alice remain incompatible. When there are different coarse-grainings, for some cases we get CH-violation, but for other cases, we do not get CH-violation. CH-violation under a particular coarse-graining signifies that the measurements are incompatible. However, when there is no CH-violation, we can not conclude anything regarding the incompatibility status.  This is depicted in TABLE \,\,\ref{table3}. \blk

 \begin{table}[h!]
     \centering
     %\scriptsize
     \begin{tabular}{|c|c|c|c|c|}
     \hline
      \multicolumn{2}{|c|}{\textbf{Possible}}  &  &  & \\
       \multicolumn{2}{|c|}{\textbf{CG}}  &  &  &\\
        \multicolumn{2}{|c|}{\textbf{choices}}  &\textbf{$\Delta S_{T}$}  &\textbf{$\Delta S_{E}$}  & \textbf{Incompatible}\\
         \multicolumn{2}{|c|}{\textbf{of}}  &  &  &\\
         
      \cline{1-2}
     
     \textbf{$A_1$} & \textbf{$A_2$} &  &  &\\
     \hline
     \hline
     (0,1) & (0,1) & 0.126 & 0.122 & yes\\
     \hline
      (0,1) & (1,2) & 0 & 0 & ?\\
      \hline
       (0,1) & (0,2) & 0.126 & 0.122 & yes\\
       \hline
     (1,2) & (0,1) & 0.126 & 0.122 & yes  \\
 \hline
 (1,2) & (1,2) & 0.126 & 0.122 & yes \\
 \hline
 (1,2) &(0,2) &0  & 0 & ?\\
 \hline
 (0,2) &(0,1) &0  & 0 & ?\\
 \hline
 (0,2) &(1,2) &0.126  & 0.122 & yes\\
 \hline
 (0,2) &(0,2) &0.126 & 0.122 & yes \\
 \hline
 
     \end{tabular}
     \caption{\label{table3}  Incompatibility status for all possible coarse-graining (CG) choices of outcomes of Alice's measurements ($A_1$ and $A_2$) in the CGLMP scenario.  Here $\Delta S_{T}$ refers to the amount of CH-inequality violation obtained in the theory and $\Delta S_{E}$ refers to the CH-inequality violation calculated from the states and measurements \textit{viz.} (\ref{expt_state}) and (\ref{expt_meas}) realised in the experiment \cite{cglmp_expt_pra24}. Here  ``yes"  denotes that under the particular CG of $A_1$ and $A_2$, the
measurements can be witnessed to be incompatible as they give CH-inequality violation. On the other hand, ``$?$" denotes that we can not make any conclusion about their incompatibility under those particular CG as there is no violation of CH-inequality in those cases. }
 \end{table}

\textit{Experimental realization.--- }  Recently, an experimental study on CGLMP inequality has been performed  \cite{cglmp_expt_pra24} where the two parties share orbital angular momentum entanglement in a scenario of multiple settings and outcomes. In their experiment, for $d=3$ they prepared an orbital angular momentum entangled state which is in $\mathbb{C}^3\otimes\mathbb{C}^3$  of the form:  
\bea \label{expt_state}
\ket{\psi}&=&0.596 \ket{+1}_A\ket{-1}_B+0.529\ket{+2}_A\ket{-2}_B \nonumber \\  && + 0.604\ket{-1}_A\ket{+1}_B,
\eea 
with $\ket{+l}_A$ corresponds to the orbital angular momentum eigenstate of the signal photon (A) with orbital angular momentum $+l \hbar \,\, \text{and}\,\, \ket{-l}_B $ corresponds the orbital angular momentum eigenstate of the corresponding idler photon (B) with orbital angular momentum $-l\hbar.$ The von Neuman measurements of Alice and Bob are given by $A_{i}\equiv\{\ket{\Gamma^i_s}_A\bra{\Gamma^i_s}\} \,\, \text{and} \,\, B_{j}\equiv\{\ket{\Theta^j_t}_B\bra{\Theta^j_t}\}$ respectively, where 

\be \label{expt_meas}
\ket{\Gamma^i_s}_A=\frac{1}{\sqrt{3}}(\ket{+1}_A+ \omega^{s+\sigma_i}\ket{+2}_A+\omega^{2(s+\sigma_i)}\ket{-1}_A),
\ee
and 
\be
\ket{\Theta^j_t}_B=\frac{1}{\sqrt{3}}(\ket{-1}_B+ \omega^{-t-\gamma_j}\ket{-2}_B+\omega^{2(-t-\gamma_j)}\ket{+1}_B),
\ee 
with $i,j\in\{1,2\} \,\, \text{and} \,\, s,t \in \{0,1,2\}, \sigma_1 = \frac{1}{4},\sigma_2 = \frac{3}{4}, \gamma_1=\frac{1}{2} \,\, \text{and} \,\, \gamma_2=0, \,\, \omega = \text{exp}(\I\frac{2\pi}{3}). $
For $d=3$ the incompatibility status of the measurements (see  (\ref{expt_meas})) used in their experimental arrangements before and after coarse-graining is depicted in TABLE \ref{table3}. 
%\blk

%\textit{Device-independent witness.---} 

\subsubsection*{Semi-device-independent witness} 
In the semi-device-independent approach in prepare-and-measure experiments, we do not have any prior knowledge of the internal functioning of the measurement device; however, we assume the dimension of the system on which the measurements act. To witness different levels of incompatibility, here we focus on a class of communication tasks, namely, $(2,\od,d)-$ random access code tasks (RAC). In this task, the sender, Alice, gets two-dit string input message $(x_1,x_2)$ with $x_1,x_2\in [\od],$ and can communicate a $d-$dimensional system to the receiver, Bob, who wants to guess the value of any of the two dits, i.e. $x_1$ or $x_2$, randomly. The notation of $(2,\od,d)-$RAC is adopted from \cite{SahaIncom}.
If any two POVMs, each with $\od$ outcomes acting on $\mathbb{C}^{d}$  are jointly measurable,  the average success probability of this task, 
\be \label{eqn_cb}
P(2,\od,d) \le P_{CB}(2,\od,d) = \frac{1}{2} \left( 1+\frac{d}{\od^2} \right),
\ee where $P_{CB}(2,\od,d)$ is an upper bound on average success probability using two compatible measurements \cite{Carmeli2020,SahaIncom}. 
It turns out any two incompatible rank-one projective measurements can be witnessed through $(2,d,d)-$RAC \cite{saha2020}. We prove the following result for full incompatibility of two 3-outcome rank-one projective measurements w.r.t. coarse-graining.
 
\begin{theorem}\label{th_cg}
Two 3-outcome rank-one projective measurements, $M=\left\{|\phi_0\rangle,|\phi_1\rangle,|\phi_2\rangle\right\}$ and $ N=\left\{|\psi_0\rangle,|\psi_1\rangle,|\psi_2\rangle\right\}$, can be witnessed to be fully incompatible w.r.t. coarse-graining via RAC  if and only if $0<|\langle \phi_{i}|\psi_{j}\rangle|<\frac{4}{5},  \quad \forall \,\, i,j =0,1,2.$ 
\end{theorem}
%\blk 

\begin{proof}

Without loss of generality, any pair of three outcome rank-one projective measurements can be written  up to unitary freedom as,
\bea 
& M=\left\{|0\rangle\langle 0|,|1\rangle\langle 1|,|2\rangle\langle 2|\right\}\equiv \left\{|\phi_{i}\rangle \langle \phi_{i}| \right\},  \nonumber   \\
& N=\left\{|\psi_{1}\rangle\langle \psi_{1}|,|\psi_{2}\rangle\langle \psi_{2}|,|\psi_{3}\rangle\langle \psi_{3}|\right\} ,
\eea 
where  $|\psi_{j}\rangle=\alpha_{j}|0\rangle + \beta_{j}|1\rangle + \gamma_{j}|2\rangle$.

Let us first consider the coarse-graining of the second and third outcomes for both measurements. So, the new measurement pair becomes: 
\bea
 &M^{(2,3)}=\big\{|0\rangle\langle 0|,|1\rangle\langle 1|+|2\rangle\langle 2|\big\}, \\\nonumber
 &N^{(2,3)}=\big\{|\psi_{1}\rangle\langle \psi_{1}|,|\psi_{2}\rangle\langle \psi_{2}|+|\psi_{3}\rangle\langle \psi_{3}|\big\}.
\eea
This pair of measurements appears in $(2,2,3)-$RAC and the compatibility bound in this case is, 
\be 
P_{CB}(2,2,3)=\frac{1}{2}\left(1+\frac{3}{2^2} \right)=\frac{7}{8}  
\ee  
by Eq.\eqref{eqn_cb} and this bound is tight. $ M^{(2,3)} \,\, \text{and} \,\, N^{(2,3)}$  will give advantage in $(2,2,3)-$RAC if
\be \label{eq_cg}
\sum_{x,y=0}^{1} ||M^{(2,3)}_{x}+N^{(2,3)}_{y}||=\sum_{i=1}^{4}||A_{i}||>7,
\ee
with 
\bea&A_{1}=M^{(2,3)}_{0}+N^{(2,3)}_{0}=\ket{0}\bra{0}+\ket{\psi_1}\bra{\psi_1}, \nonumber \\
&A_{2}=M^{(2,3)}_{0}+N^{(2,3)}_{1}= \I+\ket{0}\bra{0}-\ket{\psi_1}\bra{\psi_1}, \nonumber \\
&A_{3}=M^{(2,3)}_{1}+N^{(2,3)}_{0}=\I-\ket{0}\bra{0}+\ket{\psi_1}\bra{\psi_1}, \nonumber \\
&A_{4}=M^{(2,3)}_{1}+N^{(2,3)}_{1}=2\I-\ket{0}\bra{0}-\ket{\psi_1}\bra{\psi_1},
\eea
and $|| . ||$ denotes the maximum eigenvalue of an operator.
Now, $\ket{\psi_1}$ can be written as follows
 \bea
   &\ket{\psi_1}=\la 0|\psi_1\ra \ket{0}+\la u|\psi_1\ra \ket{u} \,\,  \text{with} \\ \nonumber 
   & \la 0|u\ra=0 \,\, \text{and} \,\, |\la 0|\psi_1\ra|^2+|\la u|\psi_1\ra|^2=1,
   \eea
   for some vector $\ket{u}.$
 So, $A_1$ can be written as a $(2\times 2)$ matrix in the basis $\{\ket{0},\ket{u}\},$  
 \begin{align}   A_1=\begin{pmatrix}
       1+|\la 0|\psi_1\ra|^2 & \la 0|\psi_1\ra \la \psi_1|u\ra \\
       \la u|\psi_1\ra \la \psi_1|0\ra & |\la u|\psi_1\ra|^2 
        \end{pmatrix},
     \end{align}
     the maximum eigenvalue of  $A_1 \,\, \text{is} \,\, 1+|\la 0|\psi_1\ra|.$ Thus, \be
     ||A_1||=1+|\la 0|\psi_1\ra|.
     \ee
Similarly, $A_2$ can be expressed in a block diagonal matrix in the ortho-normal basis $\{\ket{0},\ket{u},\ket{v}\},$  \begin{align}
    A_2=\begin{pmatrix}
        \Gamma & 0 \\
        0 & 1
    \end{pmatrix} \,\, \text{where} \,\, 
    \Gamma= \begin{pmatrix}
        2-|\la 0|\psi_1\ra|^2 & -\la 0|\psi_1\ra \la \psi_1|u\ra\\
        -\la u|\psi_1\ra \la \psi_1|0\ra  & 1-|\la u|\psi_1\ra|^2
    \end{pmatrix},
\end{align} 

and 

\be 
||A_2||=1+\sqrt{1-|\la0|\psi_1\ra|^2}. 
\ee

$A_3$ and $A_4$ can also be expressed in a block diagonal matrix in the basis $\{\ket{0},\ket{u},\ket{v}\},$  \begin{align}
    A_3=\begin{pmatrix}
        \Sigma & 0 \\
        0 & 1
    \end{pmatrix}, \ 
    \Sigma= \begin{pmatrix}
        |\la 0|\psi_1\ra|^2 & \la 0|\psi_1\ra \la \psi_1|u\ra\\
        \la u|\psi_1\ra \la \psi_1|0\ra  & 1+|\la u|\psi_1\ra|^2
    \end{pmatrix},
\end{align} 
and \be
||A_3||=1+\sqrt{1-|\la0|\psi_1\ra|^2}. 
\ee
 \begin{align}
    A_4=\begin{pmatrix}
        \Xi & 0 \\
        0 & 2
    \end{pmatrix} , \
    \Xi= \begin{pmatrix}
        1-|\la 0|\psi_1\ra|^2 & -\la 0|\psi_1\ra \la \psi_1|u\ra\\
        -\la u|\psi_1\ra \la \psi_1|0\ra  & 2-|\la u|\psi_1\ra|^2
    \end{pmatrix},
\end{align} 
and $||A_4||=2.$

Substituting the values of $||A_1||,\hdots, ||A_4||$ in Eq.\eqref{eq_cg} and after simplification we get, $|\la 0|\psi_{1}\ra|<\frac{4}{5}.$  This condition is obtained for a particular coarse-graining. Therefore, to be fully incompatible w.r.t. coarse-graining, $|\la \phi_i|\psi_{j}\ra|<\frac{4}{5} \,\, \forall \,\, i,j,$  should hold. 
%Now. from Theorem \ref{theo1} we know that for a pair of rank-one projective measurements in dimension 3, the necessary and sufficient condition for fully-incompatible w.r.t. coarse-graining is  $\la \phi_i|\psi_{j}\ra\neq 0 \,\, \forall \,\, i,j$,  this proves the Theorem \ref{th_cg}. 
\end{proof}

%%%%%%%%%%%%%%%%%%%%%%%%%%%%%%%%%%%%%%%%%%%%%%%%%%%%%%%%%%%%%%%%

\subsection{Disjoint-convex-mixing of measurements} 
We now study the operational witness of incompatibility w.r.t. disjoint-convex-mixing of measurements.

\subsubsection*{Device-independent witness}  
Given a set of $n$ measurements, we can make $k$ partitions of it,
%(with $^nC_{k}$ possible ways) 
where $k \in $ $\{2,\cdots,n\}$. If the $k$ number of measurements obtained by the disjoint-convex-mixing of the measurements provides a violation of Bell inequalities for all possible disjoint-convex-mixing and permutations, then the measurements are witnessed to be $k-$incompatible w.r.t. disjoint-convex-mixing in a device-independent way.
Using the result that any two binary-outcome incompatible measurements violate the Bell-CHSH inequality \cite{PRL103230402}, we can witness 2-incompatibility w.r.t. disjoint-convex-mixing from any set of binary-outcome measurements.

\subsubsection*{Semi-device-independent witness} 

For three or more numbers of measurements, we can witness whether the measurements are fully-incompatible w.r.t. disjoint-convex-mixture in a semi-device independent manner by constructing suitable random access code tasks.
\begin{theorem} \label{th_2}
    Three noisy Pauli measurements of Eq. \eqref{noisy_pauli} with equal noise ($\nu = \nu_0 = \nu_1 = \nu_2 $) are witnessed to be fully incompatible w.r.t. disjoint-convex-mixing via RAC if and only if $\sqrt{2/3}<\nu\leqslant 1.$ 
\end{theorem}
\begin{proof}
Consider the measurements $\{Q^{i}\}$ with $i\in \{1,2\}$ for two possible permutations of outcomes, to be formed by the disjoint-convex-mixing of $\{M\}$ and $\{N\}$  as  
\bea \label{Q1Q2}
&Q^1= \{p\,\, M_0+(1-p)\,\,N_0,\,\,p M_1+(1-p)N_1\},  \nonumber \\
&Q^2=\{p\,\,M_0+(1-p)N_1,p\,\,M_1+(1-p)N_0\}. 
\eea
To witness the incompatibility status of $(Q^i,R)$, one can construct a $(2,2,2)-$RAC task. Now, when Bob has to guess Alice's 1st bit, he performs the measurement $\{Q^i\}$ defined by \eqref{Q1Q2}, which is realized by the disjoint-convex-mixing of measurements $\{M\}$ and $\{N\}$. 
For the second bit, he performs the measurement $\{R\}$. The maximum average success probability $P(2,2,2)$ that can be obtained in the RAC task is:
%\begin{widetext}
 \bea \label{rac_cm}
P(2,2,2)&=&\frac{1}{8}\sum_{j,k=0}^1||Q^{i}_{j}+R_{k}|| \nonumber \\
&=& \frac{1}{4}(2+\sqrt{2\nu^2-2p\nu^2+2p^2\nu^2}).
\eea
%\end{widetext}
It can be shown that for all possible disjoint-convex-mixing, the maximum average success probability is the same as Eq. (\ref{rac_cm}). Now, to be fully incompatible w.r.t. disjoint-convex-mixing $ P(2,2,2)>3/4, \,\, \text{which implies} \,\, \nu^2(p^2-p+1)>\frac{1}{2}$ (by Eq.\eqref{rac_cm}).  Since the minimum value of $(p^2-p+1)$ is $\frac{3}{4},$ so for $\nu>\sqrt{2/3}$ the measurements are fully incompatible w.r.t. disjoint-convex-mixing.

In Observation \ref{r4}, we have found that this kind of three-qubit measurements with noise become fully incompatible w.r.t. disjoint-convex-mixing for $\nu >\sqrt{2/3}$. So, we can conclude that these three measurements are fully incompatible w.r.t. disjoint-convex-mixing if and only if quantum advantage is obtained in RAC. 
\end{proof}

Let us consider another example of three rank-one projective measurements acting on $\mathbb{C}^{3}$: 
\bea 
&X=\{|0_x\rangle\langle0_x|,|1_x\rangle\langle 1_x|,|2_x\rangle\langle 2_x|\}, \nonumber \\ 
&Y=\{|0_y\rangle\langle0_y|,|1_y\rangle\langle 1_y|,|2_y\rangle\langle 2_y|\},\nonumber \\  &Z=\{|0_z\rangle\langle0_z|,|1_z\rangle\langle 1_z|,|2_z\rangle\langle 2_z|\},
\eea
where, \bea
\ket{0_x}&=&\frac{1}{2}\ket{0}+\frac{1}{\sqrt{2}}\ket{1}+\frac{1}{2}\ket{2},\nonumber\\
\ket{1_x}&=&-\frac{1}{\sqrt{2}}\ket{0}+\frac{1}{\sqrt{2}}\ket{2},\nonumber\\
\ket{2_x}&=&\frac{1}{2}\ket{0}-\frac{1}{\sqrt{2}}\ket{1}+\frac{1}{2}\ket{2},\nonumber\\
\ket{0_y}&=&-\frac{\mathbbm{i}}{2}\ket{0}+\frac{1}{\sqrt{2}}\ket{1}+\frac{\mathbbm{i}}{2}\ket{2},\nonumber\\
\ket{1_y}&=&\frac{\mathbbm{i}}{\sqrt{2}}\ket{0}+\frac{\mathbbm{i}}{\sqrt{2}}\ket{2},\nonumber\\
\ket{2_y}&=&-\frac{\mathbbm{i}}{2}\ket{0}-\frac{1}{\sqrt{2}}\ket{1}+\frac{\mathbbm{i}}{2}\ket{2},\nonumber\\
\ket{0_z}&\equiv&\ket{0}, \ket{1_z}\equiv\ket{1},\ket{2_z}\equiv\ket{2}.
\eea

We take the following  disjoint-convex-mixing of the measurements $X$ and $Y$: $A = \{p|0_x\rangle\langle0_x| + (1-p) |0_y\rangle\langle0_y|, p|1_x\rangle\langle 1_x|+(1-p)|1_y\rangle\langle 1_y|, p|2_x\rangle\langle 2_x|+(1-p)|2_y\rangle\langle 2_y|\}$ with $0\leq p \leq 1$. We now consider the $(2,3,3)$ RAC game involving the two measurements $A$ and $Z$. The maximum average success probability $P(2,3,3)$ for this RAC task is given by,
%\begin{widetext}
 \be
P(2,3,3)= \frac{1}{18}  \sum_{i,j=0}^2||A_{i}+Z_{j}||,
\ee
%\end{widetext}
which turns out to be greater than $\frac{2}{3}$ i.e. $P_{CB}(2,3,3)$ for all $p\in [0,1]$. Hence, the measurements $A$ and $Z$
are incompatible. 

Next, consider the following measurement (taking an arbitrary disjoint-convex-mixing of $X$ and $Y$) $A_{i,j,k,l,m,n} = \{p|i_x\rangle\langle i_x| + (1-p) |j_y\rangle\langle j_y|, p|k_x\rangle\langle k_x|+(1-p)|l_y\rangle\langle l_y|, p|m_x\rangle\langle m_x|+(1-p)|n_y\rangle\langle n_y|\}$ with $0\leq p \leq 1$, $i,j,k,l,m,n \in \{0,1,2\}$, $i \neq k$, $k \neq m$, $m \neq i$ and $j\neq l$, $l \neq n$, $n \neq j$. Note that the aforementioned measurement $A$ is denoted by $A_{0,0,1,1,2,2}$ following the present notation. Following a similar calculation, it can be shown that any such $A_{i,j,k,l,m,n}$ is incompatible with $Z$ for all $p\in [0,1]$.

Similarly, taking an arbitrary convex combination of $X$ and $Z$ (or, $Y$ and $Z$), it can be shown that the new measurement is incompatible with $Y$ (or, $X$) for all $p\in [0,1]$.

% \begin{figure}[h!]
% 		\centering
% 		\includegraphics[scale=0.1]{rac_cm.jpg}
% 		%\caption{An unknown measurement set of arbitrary settings, $\{M_{b_{y}|y}\}_{b_y,y}$, is provided; we only know the dimension ($d$) on which this set of measurements act. Our task is to certify the incompatibility of this set of measurements. Here, the notation $[k]:= \{1,\cdots,k\}$ for any natural number $k$.}\label{fig1}
% 	\end{figure}

\textit{A generic description for the semi-device-independent witness of incompatibility.--- } Consider a measurement assemblage of $n$ measurements having $\od$ outcomes,  $\mathcal{M} = \{M_{z|x}\}_{z,x}$ where $x\in [n], z\in [\od]$, and the measurements act on $\mathbbm{C}^{d}$. We make $k$ partitions: $S=\{S_i\}_{i=0}^{k-1}$, $\bigcup_{i} S_i = \mathcal{M}, \,\text{and}\,  S_i \bigcap S_j = \phi, \,\, \forall \,\, i,j \,\, \text{with} \,\, i\neq j.$  Our purpose is to operationally witness the incompatibility of $k$ measurements produced from the disjoint-convex-mixing of $n$ measurements from each of the $k$ partitions. We can construct a $(k,\od,d)-$RAC where Alice has an $k$ dit input message, {\it viz}., $x=(x_0,x_1,\hdots,x_{k-1})$. Depending upon the message, she encodes it in a qudit and sends it to Bob, who on the other hand, gets input $y\in [k]$ and accordingly, he has to predict the value of the corresponding bit $x_y$. He performs the measurement, which is obtained by the disjoint-convex-mixing of the measurements from the partition $S_y$ and declares the outcome of the measurement. Now, if the success probability $P(k,\od,d)$ is greater than $P_{CB}(k,\od,d)$ \cite{SahaIncom} for all possible disjoint-convex-mixing and permutations of measurement outcomes, we can operationally witness $k-$incompatibility under disjoint-convex-mixing in a semi-device-independent way.

%%%%%%%%%%%%%%%%%%%%%%%%%%%%%%%%%%%%%%%%%%

%\begin{table}
	%\centering
% 	\begin{tabular}{|*{4}{c|}}
% 		\hline
% 		\multicolumn{2}{|c|}{}	& Sequential
% 		& Non-sequential  \\
% 		\multicolumn{2}{|c|}{}	& measurement & measurement \\
% 		\multicolumn{2}{|c|}{}	& 
% 		scenario & scenario  \\
% 		\hline 
% 		\multicolumn{2}{|c|}{Detectability $D$} &  $ 0.20$ & $ 0.20$ \\ 
% 		\hline 
% 		\multicolumn{2}{|c|}{Total robustness  of} & &  \\ 
% 		\multicolumn{2}{|c|}{measurement $R_{\text{total}}(M)$} & $5.06$ & $5.06$ \\ 
% 		\hline 
% 		Total & for $\rho_{W}(p_{i})$	& $1$ ebit & $\ge 1.11$ ebits \\ 
% 		\cline{2-4} 
% 		entanglement & for $\rho_{p_{i}}$	& $1$ ebit & $\ge 1.11$ ebits \\ 
% 		\cline{2-4} 
% 		consumed $ \eta $	& for $|\Psi(\theta_{i})\rangle$ & $1$ ebit & $\ge 1.11$ ebits\\ 
% 		\hline 
 	%\end{tabular}

%%%%%%%%%%%%%%%%%%%%%%%%%%%%%%%%%%%%%%%%%%%%%%%%%%%%%%%%%%%%%%%%%%%%%%%%
%%%%%%%%%%%%%%%%%  up to here %%%%%%%%%%%%%%%%%%%%%%%%%%%%%%
%%%%%%%%%%%%%%%%%%%%%%%%%%%%%%%%%%%%%%%%%%%%%%%%%%%%%%%%%%%%%%%%%%

\section{Conclusions}\label{sec7}
Measurement incompatibility is a feature in quantum theory that a set of measurements cannot be performed jointly on arbitrary systems \cite{guhne2023colloquium}. It is one of the fundamental ingredients for non-classical correlations and merits of quantum information science. Incompatibility offers a complex structure with different layers as the number of measurements and the dimension on which the measurements act increases \cite{uola21,egelhaaf_pmBell,cv_incom}. Thus, understanding the different levels of incompatibility with respect to elementary classical operations is of paramount importance from both the foundational and practical perspectives.
 In this work, we have considered the two most general classical operations
 {\it viz.}, coarse-graining of different outcomes of measurements, and disjoint-convex-mixing of different measurements. 
 
Through our present analysis, we have investigated the different levels of incompatibility arising under the above classical operations. Since
 environmental effects are ubiquitous in practical scenarios, the tolerance thresholds for maintaining measurement incompatibility against noise under classical operations are investigated here. Furthermore, we have developed a method to operationally witness different levels of measurement incompatibility in the device-independent framework involving
 Bell-type experiments, and also in the semi-device-independent framework
 involving prepare-and-measure experiments.  
 
 Several examples have been
 provided to illustrate the efficacy of the operational witnesses proposed
 here. Our results are particularly useful for the purpose of comparing measurements in terms of their degree of incompatibility. For example, measurements that remain fully incompatible with respect to coarse-graining of outcomes (or disjoint-convex-mixing of measurements) show stronger incompatibility compared to those that are not fully incompatible,  thus facilitating the legitimate choice of
measurements in an information-processing task where a high degree of incompatibility is required. Our work allows for experimental implementation and one can infer more about these layers of incompatibility by employing these operations in the statistics of previously conducted experiments \cite{cglmp_expt_pra24}.  Further, our formalism can be applied 
to predict the threshold values of system parameters enabling observations of
incompatibility layers in future experiments.

Our present study motivates future work in several open directions.
The condition for full incompatibility for two projective measurements with respect to the coarse-graining of outcomes can be extended for more
projective measurements.  Similarly, the criterion for full incompatibility of projective measurements with respect to disjoint-convex-mixing
could be generalized for any set of measurements. It will be interesting to look for examples where the full incompatibility with respect to coarse-graining can be inferred in a device-independent way from single experimental statistics. Finally, studying measurement simulability \cite{measuremestSimulability_guerini} of a given set of measurements after classical pre- and post-processing on the given set is another area for further research.

\subsection*{Acknowledgements} AKD and ASM acknowledge
support from Project No. DST/ICPS/QuEST/2018/98
from the Department of Science and Technology, Govt.
of India. DS acknowledges support from STARS (STARS/STARS-2/2023-0809), Govt.
of India. DD acknowledges the Royal Society (United Kingdom) for the support through the Newton International Fellowship (NIF$\backslash$R1$\backslash$212007). DD also acknowledges financial support from EPSRC (Engineering \& Physical Sciences Research Council, United Kingdom) Grant Numbers EP/X009467/1 and EP/R029075/1 and STFC (Science and Technology Facilities Council, United Kingdom) Grant Numbers ST/W006227/1 and ST/Z510385/1.

\bibliography{ref} 
\end{document}